\documentclass{article}
\usepackage{fullpage}
\usepackage{cite}
\usepackage{amsmath,amssymb,amsfonts}
\usepackage{algorithmic}
\usepackage{graphicx}
\usepackage{textcomp}
\def\BibTeX{{\rm B\kern-.05em{\sc i\kern-.025em b}\kern-.08em
    T\kern-.1667em\lower.7ex\hbox{E}\kern-.125emX}}
    
\usepackage{epstopdf}

\usepackage[ruled]{algorithm2e}

\SetAlFnt{\small}
\SetAlCapFnt{\small}
\SetAlCapNameFnt{\small}
\SetAlCapHSkip{0pt}
\IncMargin{-\parindent}
\usepackage{lscape}
\usepackage{graphicx}                   
\usepackage{color}
\usepackage{subfig}
\usepackage{amsmath}
\usepackage{amsthm}
\usepackage{listings}
\usepackage{epstopdf}
\usepackage{amssymb}
\usepackage{times}
\usepackage{xspace}
\usepackage{soul} 
\usepackage{enumitem}
\usepackage{url}
\usepackage{marginnote}
\usepackage{framed}
\usepackage[section]{placeins}
\usepackage{setspace}

\usepackage{ulem} 

\newcommand{\CH}{{\ensuremath{\mathcal{C}}}\xspace}
\newcommand{\VRF}{{\ensuremath{\mathcal{A}}}\xspace}
\newcommand{\Cp}{{\ensuremath{\mathit{WP_{\CH}}}}\xspace}%Penalty for the worker being caught cheating
\newcommand{\Cw}{{\ensuremath{\mathit{MP_\mathcal{W}}}}\xspace}%Penalty for the master of accepting a wrong answer
\newcommand{\Ct}{{\ensuremath{\mathit{WC_\mathcal{T}}}}\xspace}%Cost for the worker in performing the task
\newcommand{\Cv}{{\ensuremath{\mathit{MC_{\VRF}}}}\xspace}%Cost for the master to verify a task
\newcommand{\SM}{{\ensuremath{\mathit{MC_\mathcal{Y}}}}\xspace}%Cost for the Master for assigning the task to a worker
\newcommand{\SW}{{\ensuremath{\mathit{WB_\mathcal{Y}}}}\xspace}%Benefit of the worker from being assigned a task
\newcommand{\Bc}{{\ensuremath{\mathit{MB_\mathcal{R}}}}\xspace}%Benefit of the master for obtaining the right answer
\newcommand{\majoritymodel}{{\ensuremath{\cal{R}_{\rm m}}}\xspace}
\newcommand{\am}{\alpha_m\xspace}
\newcommand{\aw}{\alpha_w\xspace}

\newcommand{\typeA}{Linear\xspace}
\newcommand{\typeB}{Exponential\xspace}
\newcommand{\typeC}{Legacy Boinc\xspace}
\newcommand{\typeD}{Boinc\xspace}

\newcommand{\junk}[1]{}

\newtheorem{theorem}{Theorem}

\newtheorem{lemma}[theorem]{Lemma}

\begin{document}
%\history{Date of publication xxxx 00, 0000, date of current version xxxx 00, 0000.}
%\doi{10.1109/ACCESS.2017.DOI}

\title{Coping with Unreliable Workers in Internet-based Computing: An Evaluation of Reputation Mechanisms\thanks{This work was supported in part by the Regional Government of Madrid (CM) grant Cloud4BigData (S2013/ICE-2894) cofunded by FSE \& FEDER, 
the Spanish Ministry of Economy and Competitiveness grant FIS2015-64349-P, the European Commission FET-Open Project IBSEN (RIA~662725), the Fundaci\'on BBVA grant DUNDIG
and the NSF of China grant 61520106005.}
}

%\title{Reliability in Volunteer Computing and Crowdsourcing: A Reputation-based Reinforcement-Learning Approach}
\author{\uppercase{Evgenia Christoforou}\\
Politecnico di Torino, Torino, Italy (e-mail: evgenia.christoforou@polito.it)
\and
\uppercase{ANTONIO FERN\'ANDEZ ANTA}\\
IMDEA Networks Institute, Madrid, Spain
\and
\uppercase{CHRYSSIS GEORGIOU}\\
University of Cyprus, Nicosia, Cyprus
\and
\uppercase{MIGUEL A. MOSTEIRO}\\
Pace University, New York, USA
\and
\uppercase{ANGEL S\'ANCHEZ}\\
Universidad Carlos III de Madrid, Leganes, Madrid, Spain and \\
BIFI Institute, Zaragoza, Spain}

%\address[1]{Politecnico di Torino, Torino, Italy (e-mail: evgenia.christoforou@polito.it)}
%\address[2]{IMDEA Networks Institute, Madrid, Spain}
%\address[3]{University of Cyprus, Nicosia, Cyprus}
%\address[4]{Pace University, New York, USA}
%\address[5]{Universidad Carlos III de Madrid, Leganes, Madrid, Spain and BIFI Institute, Zaragoza, Spain}

%\markboth
%{Christoforou \headeretal: Coping with Unreliable Workers in Internet-based Computing}
%{Christoforou \headeretal: Coping with Unreliable Workers in Internet-based Computing}

%\corresp{Corresponding author: Evgenia Christoforou (e-mail: evgenia.christoforou@polito.it).}

\maketitle

\begin{abstract}
We present reputation-based mechanisms for building reliable task computing systems over the Internet. The most characteristic examples of such systems are the volunteer computing (e.g., SETI@home using the BOINC platform) and the crowdsourcing platforms (e.g., Amazon\rq{}s Mechanical Turk). In both examples end users are offering over the Internet their computing power or their human intelligence to solve tasks either voluntarily or under payment. While the main advantage of these systems is the inexpensive computational power provided; the main drawback is the untrustworthy nature of the end users. Generally, this type of systems are modelled under the ``master-worker'' setting. A ``master'' has a set of tasks to compute and instead of computing them locally she sends these tasks to available ``workers'' that compute and report back the task results. We categorize these workers in three generic types: {\em altruistic}, 
{\em malicious} and {\em rational}. Altruistic workers that always return the correct result, malicious workers that always return an incorrect result, and rational workers that decide to reply or not truthfully depending on what increases their benefit.    
We design a reinforcement learning mechanism to induce a correct behavior to rational workers, while the mechanism is complemented by four reputation schemes that cope with malice. The goal of the mechanism is to reach a state of {\em eventual correctness},
that is, a stable state of the system in which the master always obtains the correct task results.
Analysis of the system gives provable guarantees under which truthful behavior can be ensured by modelling the system as a Markov chain. 
Finally, we observe the behavior of the mechanism through simulations that use realistic system parameters values. 
%Finally applying real values to the system parameters we are able to simulate our mechanism. 
Simulations not only agree with the analysis but also reveal interesting trade-offs between various metrics and parameters. The correlation among cost and convergence time to a truthful behavior is shown and the four reputation schemes are assessed against the tolerance to cheaters.
\end{abstract}

%\begin{keywords}
%%Enter key words or phrases in alphabetical 
%%order, separated by commas. For a list of suggested keywords, send a blank 
%%e-mail to keywords@ieee.org or visit \underline
%%{http://www.ieee.org/organizations/pubs/ani\_prod/keywrd98.txt}
%Crowdsourcing, evolutionary dynamics,  reinforcement learning, reputation, volunteer computing
%\end{keywords}

%\titlepgskip=-15pt

%!TEX root = ./FGCS.tex  

\section{Introduction}
\label{sec:intro}

%\subsection{Motivation}

The Internet is turning into a massive source of inexpensive computational power. Every entity connected over the Internet is a potential source not only of machine computational power but also of human intelligence.
For this reason numerous platforms~\cite{boinc,turk,galaxyzoo,bitcoinmining} have been designed to harvest this computational power. What makes this type of computational power so appealing besides the fact that is inexpensive is the easy access that one can have to it. Unfortunately, for the same reason computations carried out over the Internet can not be trusted. Not only end users can hide behind their anonymity but also 
who ever is requesting the computation may be lacking tools to verify the validity of the results. 

Usually we refer to this source of computational power as Internet-based task computing; since we are dealing with computations carried out over the Internet in the form of small assignments for the end users. We will refer to these small assignments as tasks. 
As we mentioned before, this type of computation is inexpensive compared to supercomputing. The reason is because the end users are not devoted to the computation and they usually accept to perform these small tasks for a low payment or even volunteer for them. 
The most representative examples of Internet-based task computing are volunteer computing and crowdsourcing. 

Volunteer computing has been greatly embraced by the scientific community that is always in need of cheap supercomputing power. End users engaged by the mission of the project are willing to contribute their machine's idle computational time. The majority of these volunteering projects are using the BOINC platform~\cite{boinc}, with SETI@home~\cite{SETI} being one of the most characteristic examples. IBM's initiative through the World Community Grid~\cite{WCG} is able to bring together organizations dealing with health, poverty and sustainability with volunteers all over the Internet that want to put in a good use their idle processing power. 
Besides joining a project to support a scientific goal, a worker might also be attracted by the prestige of having her contribution announced~\cite{volunteer}. This last reason for joining the computation together with the fact that a user might actually want to harm the project are enough to jeopardize the reliability of volunteer computing. Several studies~\cite{boinc,volunteer,Heienetal09,Kondoetal2007,BOINC-WCG} have found evidence that reliability of data is not an a-priory property of volunteer computing and there is a need for establishing it.

Besides users volunteering computational resources, humans themselves connected to the Internet are a source of computational power. 
%Von Ahn pioneered    games with a purpose~\cite{von2006games} as a mean of mining useful data out of the players, while keeping them entertained through playing the online game. 
The word crowdsourcing was introduced by Howe~\cite{howe2006rise} to describe the situation where human intelligence tasks are executed over the Internet by humans that are given monetary, social or other kind of incentives.
A profit-seeking computation platform  has been developed by Amazon, called  Mechanical Turk~\cite{turk} (MTurk). Users sign in to the platform and choose to perform human intelligence tasks (HIT). In return they get a monetary reward. The most common tasks encountered in MTurk are closed class questions (following the categorization in~\cite{eickhoff2013increasing}), meaning that the range of answers is a limited predefined set.
Like volunteer computing, crowdsourcing system can not be consider reliable~\cite{eickhoff2013increasing,hyman2013software}, especially now that participants expect a monetary gain.

Another example of an application that can be considered Internet-based computing is Bitcoin mining~\cite{bitcoinmining}, that has produced a huge interest from the users and the financial industry. In Bitcoin mining workers carry out complex computations to validate transactions based on Bitcoins, a virtual currency. 
The computational paradigm is peer-to-peer, that is, there is no centralized authority. Nevertheless, the whole system can be viewed as the master assigning tasks to workers, or in this case, the miners. 
Given that miners may be deceitful, the system must include measures to prevent or minimize this drawback. However, Bitcoin relies on the complexity of the computation and Bitcoin payments (proof of work) to guarantee trustworthiness~\cite{barber}.   
 
Finally, another example of Internet-based task computing we could consider is virtual citizen science~\cite{2011galaxy,kloetzer2014learning}. We could say that this type of computation is a hybrid among what we call volunteer computing and crowdsourcing. Volunteers over the Internet willing to help scientists accept to participate in task that need human intelligence to be solved. One of the most characteristic project is Galaxy Zoo~\cite{2011galaxy,galaxyzoo}, engaging volunteers to classify galaxies into categories, doing so in many occasions in the form of a ``fun'' game.
It was through the work of Von Ahn~\cite{von2006games}, that pioneered games with a purpose, where we first saw this type of scheme for creating a more pleasant experience for the volunteers.  
As pointed out by Kloetzer et al.~\cite{kloetzer2014learning} volunteers will gradually learn to perform a task through these ``task-game mechanisms'' as they call them. Thus, these type of volunteers actually become reliable over time and given the right incentives.

%So, although the potential is great, the use of Internet-based computing is limited by the
%untrustworthy nature of the platform's components~\cite{boinc,Heienetal09}.

All of the aforementioned examples in essence follow a master-worker model. A master process sends tasks across the Internet, to available worker processes, that execute and report back the task results. Moreover, it is clear from the nature of Internet-based task computing that workers can not be trusted, for the reasons mentioned before. Thus, in order to be able to establish reliability in this type of computation we need first to be able to correctly model the workers' behavior. A number of attempts have been made in the past to classify these workers. In Distributed Computing a classical approach is to model the malfunctioning    (due to a  hardware or a software
error) or cheating (intentional wrongdoer) as {\em malicious} Byzantine workers
that wish to hamper the computation and thus always return an incorrect result. The
non-faulty workers are viewed as {\em altruistic} ones~\cite{boinc_survey} that always return the correct result. On the other hand, a game-theoretic approach assumes that workers are {\em rational}~\cite{Halp06,UDC,rational}, that is, a worker decides whether to truthfully compute and
return the correct result or return a bogus result, based on the strategy that best serves
its self-interest (increases its benefit). 
%While in~\cite{IEEETC14}, all three types were considered. 
%In this work we also make the distinction among our workers into altruistic, malicious and rational. 
Following the literature and also the motivating  examples we will categorize our workers into three types: malicious, altruistic and rational. 
Moreover we assume that the master will be interacting with the same workers for a long period of time. In the case of volunteer computing this is a natural assumption since participants usually remain in the system for a long period of time~\cite{nov2010volunteer}. On the other hand in a crowdsourcing setting it is desirable to form a group of workers that will become ``experts'' on the type of tasks requested by the master. 

Under these assumptions our goal is take advantage of the repeated interaction of the master with the workers to guarantee that at some future step of the master's interaction with the workers, the master will always be receiving the correct task result with the minimum cost.

\subsection{Background}

As part of our mechanism we use reinforcement learning to induce the correct behavior of rational workers.
Reinforcement learning~\cite{RLbook} models how system entities, or {\em learners}, interact with the environment to decide upon a strategy, and use their experience to select or avoid actions according to the consequences observed. Positive payoffs increase the probability of the strategy just chosen, and negative payoffs reduce this probability. Payoffs are seen as parameterizations of players' responses to their experiences. There are several
models of reinforcement learning. A well-known model is that of Bush and Mosteller~\cite{BM55}; this is an aspiration-based reinforcement learning model where negative effects on the probability distribution over strategies are possible, and learning does not fade with time. The learners adapt by comparing their experience with an {\em aspiration} level. In our work we adapt this reinforcement learning model and we consider a simple aspiration scheme where aspiration is fixed by the workers and does not change during the evolutionary process.

% \marginpar{Enhance related \\ work on reputation}
 
The master reinforces its strategy as a function of the reputation calculated for each worker.
Reputation has been widely considered in on-line communities that deal with untrustworthy entities, such as online auctions (e.g., eBay) or P2P file sharing sites (e.g., BitTorrent); it provides a mean of evaluating the degree of
trust of an entity~\cite{survey07}. Reputation measures can be characterized in many ways, for example, as objective or subjective, centralized or decentralized. An objective measure comes from an objective assessment process while a subjective measure comes from the subjective belief that each evaluating entity has. In a centralized reputation scheme a central authority evaluates the entities by calculating the ratings received from each participating entity. In a decentralized system entities share their experience with other entities in a distributed manner. In our work, we use the master as a central authority that objectively calculates the reputation of each worker, based on its interaction with it; this centralized approach is also used by BOINC. 
%\vspace{-1.2em} 
%

%The BOINC system itself uses a form of reputation~\cite{boinc_reputation} for an optional policy called adaptive replication. This policy avoids replication in the event that a job has been sent to a highly reliable worker.  The philosophy of this reputation scheme is to require a long time for the worker to gain a good reputation but a short time to lose it. Our mechanism differs significantly from the one that is used in BOINC. One important difference is that we use auditing to check the validity of the worker\rq{}s answers while BOINC uses only replication; in this respect, we have a more generic mechanism that also guarantees reliability of the system. Notwithstanding inspired by the way BOINC handles reputation we have designed a BOINC-like reputation type in our mechanism (called \typeC). %(see below). 
%

The BOINC system itself uses a form of reputation~\cite{boinc_reputation} for an optional policy called adaptive replication. This policy avoids replication in the event that a job has been sent to a highly reliable worker. The philosophy of this reputation scheme is to be intolerant to cheaters by instantly minimizing their reputation. Our mechanism differs significantly from the one that is used in BOINC. One important difference is that we use auditing to check the validity of the worker\rq{}s answers while BOINC uses only replication; in this respect, we have a more generic mechanism that also guarantees reliability of the system. Notwithstanding inspired by the way BOINC handles reputation we have designed a BOINC-like reputation type in our mechanism (called \typeD).
The adaptive replication policy currently used by BOINC has changed relatively recently. BOINC used to have a different policy~\cite{boinc_reputation_legacy}, where a long time was required for the worker to gain a good reputation but a short time to lose it. 
In this work we evaluate the two policies used by BOINC, adapted of course to our mechanism. We call the old policy \typeC and we seek to understand the quantitative and qualitative improvements among the two schemes.

 Sonnek et al.~\cite{sonnek07} use an adaptive reputation-based technique for task scheduling in volunteer setting (i.e., projects running BOINC).  Reputation is used as a mechanism to reduce the degree of redundancy while keeping it possible for the master to verify the results by allocating more reliable nodes. In our work we do not focus on scheduling tasks to more reliable workers to increase reliability but rather we design a mechanism that forces the system to evolve to a reliable state. We also demonstrate several tradeoff between reaching a reliable state fast and the master\rq{}s cost. We have created a reputation function (called reputation \typeA) that is analogous to the reputation function used in~\cite{sonnek07} to evaluate this function\rq{}s performance in our setting.

\subsection{Our contributions} 
%In this paper we study a master-worker crowdsourcing model where we assume the existence of a master that assigns tasks in an online fashion to a fixed set of workers. (Prediction mechanisms such as the one in~\cite{lazaro2012long} can be used to establish the availability 
%of a set of workers for a relatively long period of time.) Workers on the other hand are active and always reply to the master. Workers can be one of three types, altruistic, malicious or rationals. Our mechanisms reinforces the behavior of the rational workers by providing the appropriate incentives in order to return the correct reply. The first goal of the system is to achieve a stable state where the master will always receive the correct task reply with the minimum auditing, from there forth.  On top of that, in order to deal with malicious workers while keeping the auditing probability minimum, we implement a reputation scheme.
%Our specific contributions follow.

We aim at establishing a reliable computation system in task computing settings such as volunteer computing and crowdsourcing. This computation system is modelled as a master-worker interaction, where the master assigns tasks to a fixed set of workers in an online fashion. (Prediction mechanisms such as the one in~\cite{lazaro2012long} can be used to establish the availability 
of a set of workers for a relatively long period of time.) Workers on the other hand are active, aware of their repeated interaction with the master and willing to reply. They are categorized into three types: (1) altruistic, (2) malicious and (3) rationals. The good behavior of the rational workers is reinforced through an incentives mechanism and while the malicious workers are being identified through a number of reputation schemes proposed. The goal of the system is to achieve  a stable state where the master will always receive the correct task reply with the minimum cost to the master (i.e. auditing), from there forth. In detail, our contributions are as follows.

\begin{itemize}[leftmargin=4mm]   
\item We design such %(in relation to previous sentence) 
an algorithmic mechanism that uses reinforcement learning (through reward and optional punishment) to induce a correct behavior to rational workers while coping 
with malice using {\em reputation}. 
%\ecc{
\item We consider a centralized reputation scheme controlled by the master that
may use four different reputation metrics to calculate each worker's reputation. The
first (reputation type \typeA) is adopted from~\cite{sonnek07} and it is a simple approach for calculating reputation. The second reputation type (called \typeB) , which we introduce, allows for a more drastic change of reputation simply because the mathematical function used changes faster. 	The third reputation type is inspired by BOINC's current reputation scheme~\cite{boinc_reputation}, thus we call it \typeD. Finally, the fourth reputation scheme~\cite{boinc_reputation} is inspired by the previously used reputation scheme~\cite{boinc_reputation_legacy} of BOINC (called \typeC) and it uses an 
indirect way of calculating reputation through an error rate.
%}
%We consider a centralized reputation scheme controlled by the master that may use three different reputation metrics to calculate each worker's reputation. The first is adopted from~\cite{sonnek07}, the second, \cg{which we introduce}, allows for a more drastic change of reputation and the third is inspired by BOINC\rq{}s reputation scheme~\cite{boinc_reputation}.
\vspace{.2em}
\item We analyze our reputation-based mechanism modeling it as a Markov chain and we identify conditions under which truthful behavior can be ensured. We analytically prove that by using the reputation type \typeB (i.e. the one we introduce) reliable computation is eventually achieved.\vspace{.2em}
\item Simulation results, obtained using parameter values extracted by BOINC-operated applications (such as \cite{emBoinc,einstein} ), reveal interesting
trade-offs between various metrics and parameters, such as cost, time of convergence to a truthful behavior, tolerance to cheaters and the type of reputation metric employed. Simulations also reveal better performance of our reputation type (\typeB) in several realistic cases.
\end{itemize}

%\vspace{-1em} 

\subsection{Related work}  

%A classical Distributed Computing approach is
% to model the malfunctioning (due to a  hardware or a software
%error) or cheating (intentional wrongdoer) as {\em malicious} workers
%that wish to hamper the computation and thus always return an incorrect result. The
%non-faulty workers are viewed as {\em altruistic} ones~\cite{boinc_survey} that always return the correct result. 
%Under this view, 
%malicious-tolerant protocols have been considered, e.g.,~\cite{Sarmenta02,ALEX,PPL12}, where the master 
%decides on the correct result based on majority voting.
% A Game-theoretic approach is
%to assume that workers are {\em rational}~\cite{Halp06,UDC,rational}, that is, a worker decides whether to truthfully compute and
%return the correct result or return a bogus result, based on the strategy that best serves
%its self-interest (increases its benefit). Under this view, incentive-based algorithmic mechanisms have been
%devised, e.g.,~\cite{CCS,NCA08}, that employ reward/punish schemes to ``enforce'' rational workers
%to act correctly. 

In Distributed Computing intentional or unintentional misbehaviours (a hardware or software error) from the workers are modelled assuming Byzantine (malicious) workers that want to harm the computation and thus reply with an incorrect result. Workers that are not misbehaving, that is they reply with a correct result are seen as altruistic workers.       Under this view, 
malicious-tolerant protocols have been considered, e.g.,~\cite{Sarmenta02,ALEX,PPL12}, where the master 
decides on the correct result based on majority voting.
On the other hand, game-theoretic approaches have been devised modelling the workers as {\em rational}~\cite{Halp06,UDC,rational}, that will compute and reply with the correct result only if this strategy maximizes their benefit (utility), otherwise they will choose a strategy where they are dishonest, replying with an incorrect result. Under this view, incentive-based algorithmic mechanisms have been
devised, e.g.,~\cite{CCS,NCA08}, that employ reward/punish schemes to ``enforce'' rational workers
to act correctly.

In~\cite{IEEETC14}, all three types were considered, and both approaches were combined in order to produce an algorithmic mechanism that provides incentives to rational workers to act correctly, while alleviating the malicious workers' actions. All the solutions described are {\em one-shot (or stateless)} in the sense that the master decides about the outcome of an interaction with the workers involving a specific task, without using any knowledge gained by prior interactions. In~\cite{CCPE13}, we took advantage  of the repeated interactions between the master and the workers, assuming the presence of {\em only} rational workers.  For this purpose, we studied the {\em dynamics of evolution}~\cite{MSJ82} of such master-worker computations through {\em reinforcement learning}~\cite{RLbook} where both the master and the workers 
adjust their strategies based on their prior interaction. The objective of the master is to reach 
a state in the computation after which it always obtains the correct results, while the workers attempt to increase their benefit. 
Hence, prior work either considered all three types of workers in one-shot computations, or multi-round interactions assuming only rational workers.
%%\vspace{-.4em}
%
%----------NOT INCLUDED
%In volunteer computing workers join projects to support a scientific goal and/or to gain prestige~\cite{volunteer}, while in non-volunteer computing workers expect payment. Whatever the reason,  Internet-based computing can not be considered a reliable platform~\cite{boinc,volunteer,Heienetal09,Kondoetal2007,eickhoff2013increasing,BOINC-WCG}. Thus provable guarantees must be given that the designed mechanism provides a reliable platform, especially in commercial platforms where one can not consider altruistic workers. The existence of all three types of workers must be assumed since workers can have a predefined behavior (malicious or altruistic) or not (rational) as we observe from the survey conducted by SETI@home~\cite{boinc_survey}, the behavior of its users~\cite{boinc_stats} 
%and studies conducted in crowdsourcing platforms~\cite{eickhoff2013increasing,hyman2013software}. 
%A mechanism must be designed that benefits from the repeated interaction with the workers and thus detaches the knowledge of the distribution over the type of workers from the assumptions (in comparison with~\cite{IEEETC14}). 
%----------END

%A few things about BAR
Aiyer et al.~\cite{BAR} introduce the BAR model to reason
about systems with Byzantine (malicious), Altruistic, and
Rational participants. They also introduce the notion
of a protocol being BAR-tolerant, that is, the protocol
is resilient to both Byzantine faults and rational manipulation. As an application, they
designed a cooperative backup service for P2P systems, based on
a BAR-tolerant replicated state machine. Li et al~\cite{BAR-gossip}
also considered the BAR model to design a P2P live streaming application
based on a BAR-tolerant gossip protocol. Both works employ incentive-based
game theoretic techniques (to remove the selfish behavior), but the
emphasis is on building a reasonably practical system (hence, formal
analysis is traded for practicality). 
Recently, Li et al~\cite{BAR-FP} developed a P2P streaming application,
called FlightPath, that provides a highly reliable data stream to a dynamic 
set of peers. FlightPath, as opposed to the above-mentioned BAR-based works,
is based on mechanisms for {\em approximate equilibria}~\cite{CS_SODA07},
rather than strict equilibria. In particular, $\epsilon$-Nash equilibria
are considered, in which rational players deviate if and only if
they expect to benefit by more than a factor of $\epsilon$. As the authors
claim, the less restrictive nature of these equilibria enables the
design of incentives to limit selfish behavior rigorously, while 
it provides sufficient flexibility to build practical systems.
 More recent works have considered other problems in the BAR model (e.g., data transfer~\cite{BARTransfer}). Although the objectives and the model considered are different, our reputation-based mechanism can
be considered, in some sense, to be BAR-tolerant.

Various surveys focus on the obstacles and challenges that current crowdsourcing approaches face~\cite{yuen2011survey,kittur2013future,silberman2010ethics}.
One of the most crucial issues that crowdsourcing faces is the cheating behaviour of workers. 
In~\cite{zhang2012reputation} reputation-based incentive protocols are presented. The analysis is done over the presence of many masters that are matched with workers. The task is assigned to a single worker and payments are ex-ante (i.e., the master pays before the worker performs the task). They design an algorithm that prevents workers from not putting any effort in performing the task (i.e., this can be perceived as cheating). 
Eickhoff and de Vries~\cite{eickhoff2013increasing} investigate the nature and the causes for the cheating behavior of workers in crowdsourcing platforms. They categorize Human Intelligence Tasks (HITs) in two categories: (1) \emph{closed class questions} where the workers selects the correct answer from a limited list of options, and (2) \emph{open class questions} where the answer is not a strict list and/or the task is of a creative nature. In their work they propose ways to discourage cheaters from performing HITs by transforming the task to be less appealing to them. Our approach also maintains a reputation for each worker (in order to know how reliable his answers are),
but it does not discourage workers from participating. Instead, it tries to reinforce their good behavior to use them in future computations. We do not consider here the problem of matching masters and workers, which we assume solved, and focus on the problem with one master and multiple workers, all of them used by the master. 

%Aiyer et al.~\cite{BAR} introduced the BAR model to reason about systems with Byzantine
%(malicious), Altruistic, and Rational participants. They introduced the notion of {\em BAR-tolerant} protocols, i.e., protocols that are resilient to both Byzantine faults and rational manipulation. As an application, they designed a cooperative backup service for P2P systems, based on a BAR-tolerant replicated state machine.

%\vspace{-.5em} 

%placed in a different file, so to easily swap between long and short version
%\input{rwork}

\section{Model}
\label{rep:model}

In this section we characterize our model and we present the concepts of auditing, payoffs, rewards and aspiration. We also give a formal definition of the four reputation types used by our mechanism.  
%\vspace{-1em} 
\paragraph{\bf \em Master-Worker Framework}
\noindent
%\ecc{
We consider a system consisting of a set $G$ of voluntary workers. The set $G$ is broken down into disjoint sets $W_j$ of size $n$ forming the group of workers receiving a replica of the same tasks. For simplicity we will focus at only one such set of workers named $W$. Hence we consider a master and a set $W$ of $n$ workers.
%}
%\st{We consider a master and a set $W$ of $n$ workers.} 
The computation is broken into
%\deleted[M][if we want to say asynchronous we need to explain in what sense]{(asynchronous)} 
{\em rounds}, and in each round the master sends a task to the workers to compute and return 
the result. Based on the workers' replies, the master must decide which is the value most likely to be the correct result for this round. Crowdsourcing applications provide the possibility of selecting workers~\cite{turk,microworkers}.
\paragraph{\bf \em Tasks}\noindent From the perspective of crowdsourcing, tasks are closed class questions, whereas from the perspective of BOINC-operated applications, they are computations. These tasks 
have a unique solution; although such limitation reduces the scope of application of the 
presented mechanism~\cite{TACB05}, there are plenty of computations where the correct 
solution is unique: e.g., any mathematical function. %\vspace{-.9em} 

\paragraph{\bf \em Worker types} \noindent We consider that workers can be categorized in three types: {\em rational, altruistic} and {\em malicious}.
Rational workers are selfish in a game-theoretic sense and their aim is to maximize their utility (benefit). In the context of this paper, a worker is {\em honest} 
in a round, when it truthfully computes and returns the correct result, and it {\em cheats} when it returns 
some incorrect value. Altruistic and malicious workers have a predefined behavior, to always be honest or cheat, respectively. Instead, a rational worker decides to be honest or cheat depending on which strategy maximizes its utility. We denote by $p_{Ci}(r)$ the probability of a rational worker $i$ cheating in round $r$. This probability is not fixed and the worker adjusts it over the course of the computation. The master is not
aware of the worker types, neither of a distribution of types (our mechanism does not rely on any statistical information). %However, it is reasonable to assume that not all workers are malicious.

%\junk{(WHAT ABOUT THIS INSTEAD:)
%We logically assume the existence of two categories of workers (as in societies): {\em with and without predefined behavior}. Workers with a predefined behavior do not consider (are senseless) to the payoff they receive from the system, they are subject only to their own motivations. In this category we have the malicious workers, that always cheat (due to a hardware/software error or deliberately); and the altruistic workers that are always honest (predisposed to aid).  
%On the other hand workers with no predefined behavior are considered rational i.e. their behavior is guided by the payoffs they receive form the system. In a game-theoretic view their aim is to maximize their utility (benefit).
%% under the assumption that other workers do the same.
%A rational worker decides to be honest or cheat depending on which strategy maximizes its utility. We denote by $p^{(r)}_{Ci}$ the probability of a rational worker $i$ cheating in round $r$. This probability is not fixed and the worker adjusts it over the course of the computation. The master is not aware of the worker types, neither of a distribution of types (our mechanism does not rely on any statistical information). 
%}

While workers make their decision individually and with no coordination, following \cite{Sarmenta02} and \cite{PPL12}, we assume that all the workers that cheat in a round  return the same incorrect value; this yields 
a worst case scenario (and hence analysis) for the master with respect to obtaining the correct result using mechanisms where the result is the outcome of voting. It subsumes models where cheaters do not 
necessarily return the same answer. (This can be seen as a weak form of collusion.)

%\evgenia{Empirical evaluations of SETI-like systems, reported in~\cite{einstein} and \cite{emBoinc}, suggest that malicious workers are usually no more than 10\%. A survey by BOINC supports these results~\cite{boinc_survey}.}
For simplicity, unless otherwise stated, we assume that workers do not change their type over time.
In practice it is possible that changes occur.
For example, a rational worker might become malicious due to a bug, or a malicious worker
(e.g., a worker under the influence of a virus) become altruistic (e.g., if an antivirus software reinstates it). 
If this may happen, then all our results still apply for long enough periods between two changes. (In our simulation
study we do consider scenarios where the workers change their type dynamically.%\vspace{-.9em} 
%our results the whole evolution of the 
%To keep the discussion simple, and the analysis feasible, we focus on (long) periods of time in which each of the $n$ workers has the same type in the considered period. We show that eventual correctness is reached in a few rounds at the beginning of that period, hence, if the behavior (type) of certain workers changes, and eventual correctness is lost, this can be seen as the beginning of another period, in which
%eventual correctness will again be reached.

\paragraph{\bf \em Auditing, Payoffs, Rewards and Aspiration}
\noindent
To induce the rational workers to be honest, the master employs, when necessary, {\em auditing} and
\emph{reward/punish} schemes. The master, in a round, might decide to audit the response of the workers, 
at a cost. % \evgenia{(that is eventually minimized)}. 
In this work, auditing means that the master computes the task by itself, and checks
which workers have been honest. We denote by $p_\VRF(r)$ the probability of the master auditing the 
responses of the workers in round $r$. The master can change this auditing probability over the course of the computation,
but restricted to a minimum value $p_\VRF^{min}>0$.
When the master audits, it can accurately reward and punish workers. 
When the master does not audit, it rewards only those in the weighted majority (see below)
%weighted by reputation as discussed in the next section) 
of the replies received and punishes no one. 
\noindent
In this work we consider three worker payoff parameters: 
\vspace{-.5em}
\begin{itemize}
\item[(a)]$\Cp$: worker's punishment for being caught cheating,
\item[(b)]$\Ct$: worker's cost for computing a task
\item[(c)]$\SW$: worker's benefit (typically payment) from the
master's reward. 
\end{itemize}  
Also, following~\cite{BM55}, we assume that, in every round, a worker $i$ has an 
{\em aspiration} $a_i$: the minimum benefit it expects to obtain in a round. 
In order to motivate the worker to participate in the computation, the master usually ensures that 
$\SW \geq a_i$; in other words, the worker has the potential of its aspiration to
be covered. We assume that the master knows the aspirations.
%\sout{This information can be included, for example, in a contract the master and the
%worker agree on, prior to the start of the computation.}
Finally, we assume that the master has the freedom of choosing $\SW$ and $\Cp$ with 
goal of eventual correctness.%\vspace{-.8em} 
%\vspace{-.7em}
%; by tuning these parameters
%and choosing $n$, the master can achieve the 
%All other parameters can either be fixed because they are system parameters 
%or may also be chosen by the master (except the aspiration, which is a parameter
%set by each worker).}

\paragraph{\bf \em Eventual Correctness} \noindent
The goal of the master is to eventually obtain a reliable computational platform: After some 
finite number of rounds, the system must guarantee that the master obtains the correct task results 
in every round with probability $1$ and audits with probability $p_\VRF^{min}$. 
We call such property \emph{eventual correctness}.%\vspace{-.9em}

\paragraph{\bf \em Reputation} \noindent
%\marginpar{\ec{mention that \\ the way we \\ aggregate \\ reputation\\ is through adding\\ the values of the \\ workers.}}
The reputation of each worker is measured by the master; a centralized reputation mechanism is used. 
In fact, the workers are unaware that a reputation scheme is in place, and their interaction with 
the master does not reveal any information about reputation; i.e., the payoffs do not depend on a worker's reputation. 
%Therefore, reputation is only used to assist the master in obtaining the correct result rather 
%than to provide more incentives for workers to be truthful (that would be the case when payoffs would depend on %reputation).

In this work, we consider {\em four} reputation metrics. 
The first one is analogous to a reputation metric used in~\cite{sonnek07}  and we call it {\em \typeA} in this work. 
Reputation {\em \typeD} is inspired by the BOINC's adaptive replication metric currently in use~\cite{boinc_reputation}, while reputation metric {\em \typeC}
is inspired by the previously used version of BOINC's adaptive replication metric~\cite{boinc_reputation_legacy}.  
We present the performance of our system under both BOINC reputation metrics as an opportunity to compare and contrast these two schemes within our framework. 
%We decided to present the performance of our system under 
% this reputation metric as well for the sake of experimentation and for consistency with previous works. 
Finally, the last reputation metric we consider is reputation {\em \typeB} that is not influenced by any other reputation type, and as we show in Section~\ref{rep:analysis} it possesses beneficial properties.  
%In this work, we consider four reputation metrics. The first one, called {\em \typeA} is analogous to a reputation metric used in~\cite{sonnek07} and the third one,  called {\em \typeC} is inspired by BOINC~\cite{boinc_reputation}. We also define our own type called {\em \typeB}  that is not influenced by any other reputation type, and as we show in Section~\ref{rep:analysis} it possesses beneficial properties.  
 In  all  types, the reputation of a worker is determined based on the number of times it was found truthful. Hence, the master may update the reputation of the workers only when it 
audits. We denote by $aud(r)$ the number of rounds the master audited up to round $r$, and by $v_i(r)$ we refer to the number of auditing rounds in which worker $i$ was found truthful up to round $r$. 
Moreover, we define $streak_i(r)$ as the number of rounds $\leq r$ in which worker $i$ was audited, and replied correctly after the latest round in which it was audited, and caught cheating. 
We let $\rho_i(r)$ denote the {\em reputation} of worker $i$ after round $r$,  and for a given set of workers $Y\subseteq W$ we let $\rho_Y(r)=\sum_{i \in Y} \rho_i(r)$ be the aggregated reputation of the workers in $Y$, by aggregating we refer to summing the reputation values.
Then, the reputation types we consider are detailed in Figure~\ref{definition:reputation}.

\begin{figure}[t]
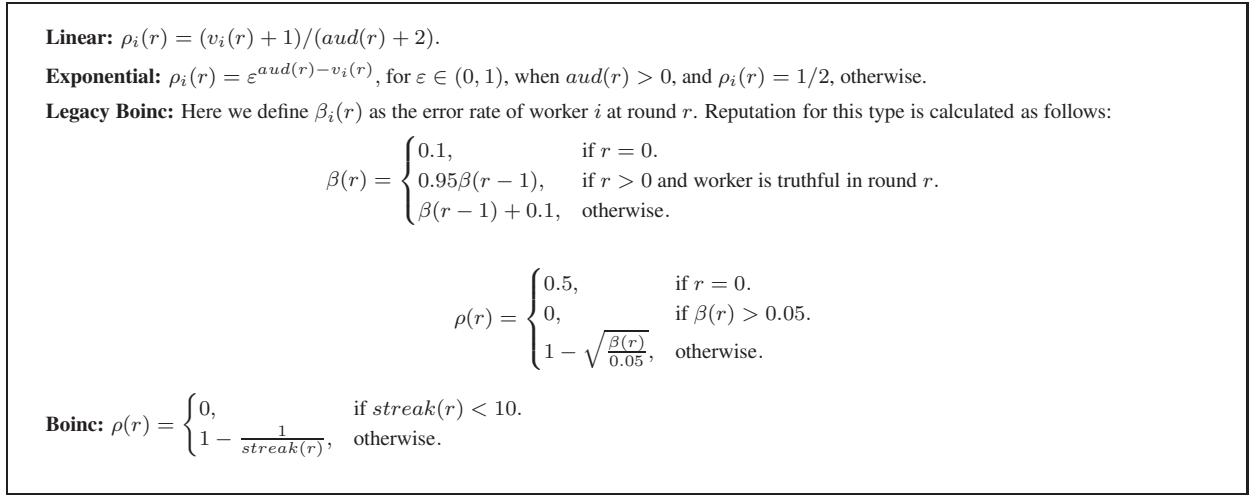

\begin{framed}
\begin{footnotesize}
\begin{itemize}[leftmargin=2mm]
\item [] {\bf \typeA:} 
$\displaystyle
\rho_i(r) = (v_i(r)+1)/(aud(r)+2).$ %\vspace{-.1em}
\item [] {\bf \typeB:} 
%When $aud(r)=0$  then $\rho_i(r) =1/2$, otherwise, $\displaystyle \rho_i(r) =  \varepsilon^{aud(r)-v_i(r)},$\\
$\displaystyle
\rho_i(r) =  \varepsilon^{aud(r)-v_i(r)},$ for $\varepsilon \in (0,1)$, when $aud(r)>0$, and 
$\displaystyle 
\rho_i(r) =1/2$, otherwise.
%\text{(b)} & \rho_i =  \varepsilon^{1-correct\_audit_i(r)/audit(r)}
%\text{(c)} & \rho_i =  \varepsilon^{1-\frac{correct\_audit_i(r)+1+1/\lg \varepsilon}{audit(r)+1}}\\
\item [] {\bf \typeC:} Here we define $\beta_i(r)$ as the error rate of worker $i$ at round $r$. 
%and by $A=0.05$ the error bound. The error rate $\beta(\cdot)$ is initialized as $\beta_i(0)= 0.1$. 
Reputation for this type is calculated as follows: \vspace{-1.5em}

\begin{align*}
  \beta(r)=\begin{cases}
    0.1, & \text{if $r=0$}.\\
    0.95\beta(r-1), & \text{if $r>0$ and worker is truthful in round $r$}.\\
    \beta(r-1)+0.1, & \text{otherwise}.
  \end{cases}\end{align*}

\begin{align*}
  \rho(r)=\begin{cases}
    0.5, & \text{if $r=0$}.\\
    0, & \text{if $\beta(r)>0.05$}.\\
    1- \sqrt{\frac{\beta(r)}{0.05}}, & \text{otherwise}.
  \end{cases}\end{align*}
  \item [] {\bf \typeD:} 
$
  \rho(r)=\begin{cases}
    0, & \text{if } streak(r)<10.\\
    1-\frac{1}{streak(r)}, & \text{otherwise}.
  \end{cases}$
\end{itemize}
\end{footnotesize}
\end{framed}
\caption{Reputation types.}
\label{definition:reputation}
\end{figure}

These four reputation types satisfy the following two natural properties: 

\begin{itemize}
\item[A.] If a worker is honest when the master audits, the reputation of the worker cannot decrease.
\item[B.] If a worker cheats when the master audits, the reputation of the worker cannot increase.
\end{itemize}

This claim is proven below in Lemmas~\ref{lemma1} and~\ref{lemma2}.

\begin{lemma}
\label{lemma1}
Natural Property A holds for reputation type \typeA , \typeB , \typeC and \typeD. 
\end{lemma}

\begin{proof} We present separately the proof for each reputation type. \\
{\it \typeA:} 
Assume that at state $s_r$ worker $i$ has reputation $\rho_i(r)= \frac{v_i(r)+1}{aud(r)+2}$ and in the next state the master audits and the worker is honest then reputation becomes $\rho_i(r+1)= \frac{v_i(r)+2}{aud(r)+3}$. Since $\frac{v_i(r)+1}{aud(r)+2} \leq \frac{v_i(r)+2}{aud(r)+3}$ the lemma holds for reputation type \typeA. \\
{ \it \typeB:}
Assume that at state $s_r$ worker $i$ has reputation $\rho_i(r)=\varepsilon^{aud(r)-v_i(r)}$ and in the next state the master audits and the worker is honest then reputation becomes $\rho_i(r+1)=\varepsilon^{aud(r)-v_i(r)}$, then the lemma trivially holds for reputation type \typeB.
\\ 
{ \it \typeC:} 
At state $s_r$ worker $i$ can have reputation $\rho_i(r)\geq 0$ and in the next state the master audits and the worker is honest then $\beta_i(r)< 0.05$ and $\rho_i(r+1)\geq 0$. If in state  $s_r$ $\rho_i(r)=0$ then the natural property holds. If in state $s_r$, $\rho_i(r)=1-\sqrt{\frac{\beta_i(r)}{0.05}}$ then in the next state $\rho_i(r+1)=1-\sqrt{\frac{\beta_i(r+1)}{0.05}}$. The property still holds since $\beta_i(r)<\beta_i(r+1)$ and thus the claim is proved for reputation type \typeC. \\
{\it \typeD:} 
Three possible cases exist:
1)Assume that at state $s_r$ worker $i$ has reputation $\rho_i(r)= 0$ and $streak_i(r)<9$ and in the next state the worker is honest, then reputation becomes $\rho_i(r+1)= 0$. Since $\rho_i(r)=\rho_i(r+1)$ the lemma holds for this case. 
2)Assume that at state $s_r$ worker $i$ has reputation $\rho_i(r)=0$ and $streak_i(r)=9$ and in the next state the worker is honest, then reputation becomes $\rho_i(r+1)=\frac{9}{10}$. Since $0 \leq \frac{9}{10}$ the lemma holds for this case too.
3)Finally, assume that at state $s_r$ worker $i$ has reputation $\rho_i(r)=1-\frac{1}{streak_i(r)}$ and $streak_i(r)\geq 10$ and in the next state the worker is honest, then reputation becomes   $\rho_i(r+1)=1-\frac{1}{streak_i(r+1)}$. Since $streak_i(r)<streak_(r+1)$ then $1-\frac{1}{streak_i(r)} \leq 1-\frac{1}{streak_i(r+1)}$ and the lemma holds for this case. Thus, the lemma for reputation type \typeD holds since it is true for all three possible cases. 
\end{proof}

\begin{lemma}
\label{lemma2}
Natural Property B holds for reputation type \typeA, \typeB, \typeC and \typeD. 
\end{lemma}

\begin{proof}
We present separately the proof for each reputation type. \\
{\it \typeA:} 
Assume that at state $s_r$ worker $i$ has reputation $\rho_i(r)= \frac{v_i(r)+1}{aud(r)+2}$ and in the next state the master audits and the worker cheats then reputation becomes $\rho_i(r+1)= \frac{v_i(r)+1}{aud(r)+3}$. Since $\frac{v_i(r)+1}{aud(r)+2} \geq \frac{v_i(r)+1}{aud(r)+3}$ the lemma holds. \\
{\it \typeB:} 
Assume that at state $s_r$ worker $i$ has reputation $\rho_i(r)=\varepsilon^{aud(r)-v_i(r)}$ and in the next state the master audits and the worker cheats then reputation becomes $\rho_i(r+1)=\varepsilon^{aud(r)+1-v_i(r)}$. Since $\varepsilon \in (0,1)$ then $\varepsilon^{aud(r)-v_i(r)} \geq \varepsilon^{aud(r)+1-v_i(r)}$ and the lemma holds. \\
{\it \typeC:} 
At state $s_r$ worker $i$ can have reputation $\rho_i(r)\geq 0$ and in the next state the master audits and the worker cheats then $\beta_i(r+1)>0.05$ and $\rho_i(r+1)=0$ and the claim holds.  \\
{\it \typeD:} 
Two possible cases exist: 1)Assume that at state $s_r$ worker $i$ has reputation $\rho_i(r)= 0$ and $streak_i(r)<10$ and in the next state the master audits and the worker cheats, then  $streak_i(r+1)=0$ and reputation becomes $\rho_i(r+1)= 0$. Since $\rho_i(r) = \rho_i(r+1) $ the lemma holds for this case. 2)Assume that at state $s_r$ worker $i$ has reputation $\rho_i(r)=1-\frac{1}{streak_i(r)} $ and $streak_i(r) \geq 10$ and in the next state the master audits and the worker cheats. Then $streak_i(r+1)=0$ and the reputation becomes  $\rho_i(r+1)= 0$. Since, $1-\frac{1}{streak_i(r)}>0$ the lemma hold also for this case. Thus, the lemma holds in general.  
\end{proof}

In each round, when the master {\em does not audit}, the result is obtained from the {\em weighted majority} as follows. 
Consider a round $r$. Let $F(r)$ denote the subset of workers that returned an incorrect result, i.e., the rational workers who chose to cheat plus the malicious ones; recall that we assume as a worst case that all cheaters return the same value. Then, $W\setminus F(r)$ is the subset of workers that returned the correct value, i.e., the rational workers who chose to be truthful plus the altruistic ones. Then, if $\rho_{W\setminus F(r)}(r) > \rho_{F(r)}(r)$, the master will accept the correct value, otherwise it will accept an incorrect value. 
The  mechanism, presented in the next section, employs auditing and appropriate incentives so 
that rational workers become truthful and have a reputation that is higher than that of the malicious workers.

\section{Reputation-based Mechanism}
\label{rep:mechanism}

We now present our reputation-based mechanism. 
%Reputation is only used by the master; the workers are not aware of the reputation 
%mechanism and their payoffs do not depend on it. 
The mechanism is composed by an algorithm run by the master and an algorithm
run by each worker.

\paragraph{\bf \em Master's Algorithm} \noindent
The algorithm begins by choosing the initial probability of auditing and the initial reputation (same for all workers). The initial probability of auditing will be set according to the information the master has
about the environment (e.g., workers' initial $p_C$). For example, if it has no information about the environment, a possibly safe approach is to initially set $p_\VRF=0.5$. The master also chooses the reputation type
to use. 

After that, at each round, the master sends a task to all workers and, when all answers are received, 
the master audits the answers with probability $p_\VRF$. In the case the answers are not audited, the master accepts the value returned by the weighed majority,
%, that is, the workers with the higher accumulated reputation value (cf. Section~\ref{rep:model}).} 
and continues to the next round with the same probability of auditing and the same reputation values for each worker. In the case the answers are audited, 
%{\color{magenta} the reputation of each worker changes according to its reply and} 
the value $p_\VRF$ of the next round is reinforced (i.e., modified according to the accumulated reputation of the cheaters) and the reputations of the workers are updated based on their responses.
%and the reputation type used (type 1 or 2 as defined in Section~\ref{rep:model}).} 
Then, the master rewards/penalizes the workers appropriately. Specifically, if the master audits and a worker $i$ is a cheater (i.e., $i\in F$), then $\Pi_i = -\Cp$; if $i$ is honest, then $\Pi_i = \SW$. If the master does not audit, and $i$ returns the value of the weighted majority (i.e., $i\in W_m$), then
$\Pi_i = \SW$, otherwise $\Pi_i=0$.
%A pseudocode of the algorithm is given in Algorithm~\ref{alg3}. 
%\marginpar{\miguel{moved down}}

\lstset{
columns=fullflexible,
tabsize=3,
basicstyle=\small,
identifierstyle=\rmfamily\textit,
mathescape=true,
morekeywords={to,process,const,when,elsif,procedure,null,function,do,return,foreach,begin,var,if,then,else,rcv,rx_user_events,rx_network_events,
for,end,receive,send,upstream,downstream,while,do,forward,backward,upon,wait,audit,update,accept},
literate={:=}{{$\leftarrow$ }}1{->}{{ $\rightarrow$ }}1,
escapeinside={('}{')}}
\begin{algorithm}[t]
\begin{lstlisting}
$p_\VRF$ := $x$, where $x \in [p_\VRF^{min}, 1]$
$aud = 0$
// initially all workers have the same reputation
$\forall i\in W: v_i = 0; \beta_i=0.1; \rho_i = 0.5$; $streak_i=0$ 			
for $r$ := $1$ to  $\infty$ do
  send a task $T$ $\mathit{to}$ all workers in $W$
  upon receiving all answers do
    audit the answers with probability $p_\VRF$
    if the answers were not audited then
      // weighted majority, coin flip in case of a tie
      accept the value returned by workers in $W_m\subseteq W,$
            where $\rho_{W_m}>\rho_{W\setminus W_m}$  			
    else 		 		// the master audits
      $aud$ := $aud + 1$	      
      Let $F\subseteq W$ be the set of workers that cheated.
      $\forall i\in W: $  
         // honest workers   	
      	if $i \notin F$ then $v_i$ :=$v_i + 1$ or $\beta_i$:= $\beta_i \cdot 0.95$ or 
	                   $streak_i$:=$streak_i+1$
      	 // cheater workers
      	 else $v_i$ :=$v_i$ or $\beta_i$:= $\beta_i + 0.1$ or $streak_i$:=$0$
      	update reputation $\rho_i$ of worker $i$ as defined 
      	       by reputation type used		
      if ${\rho_W} = 0$ then $p_\VRF$ := $\min\{1, p_\VRF+\am\}$  else
        $p_\VRF$ := $\min\{1, \max\{p_\VRF^{min}, p_\VRF+\am(\frac{\rho_F}{\rho_W}-\tau) \}\}$ 
    $\forall i\in W :$ return payoff $\Pi_i \mathit{to}$ worker $i$ 
\end{lstlisting}
\caption{Master's Algorithm}
\label{alg3}\vspace{-.5em}
\end{algorithm}
%\end{footnotesize}
%\end{figure*}

%\begin{figure*}[ttt!]
%\begin{footnotesize}
 \begin{algorithm}[t]
\begin{lstlisting}
$p_{Ci}$ := $y$, where $y \in [0, 1]$
for $r $ := $1$ to $\infty$ do
  receive a task $T$ from the master
  $S_i$ := $-1$ with probability $p_{Ci}$, 
  and $S_i$ := $1$ otherwise
  if $S_i = 1$ then
      $\sigma$ := $compute(T)$, 
  else
      $\sigma$ := arbitrary solution
  send response $\sigma$ $\mathit{to}$ the master
  get payoff $\Pi_i$
  if $S_i = 1$ then
     $\Pi_i$ := $\Pi_i - \Ct$ 
  $p_{Ci}$ := $\max\{0,\min\{1,p_{Ci} - \aw(\Pi_i-a_i)S_i\}\}$
\end{lstlisting}
\caption{Algorithm for Rational Worker $i$}
\label{alg2}\vspace{-.5em}
\end{algorithm}
%\hfill
%\end{footnotesize}
%\end{figure*}

\junk{%redundant
Given that, when the answers are not audited, the master has no information about the number of cheaters in the round, the probability $p_\VRF$ remains the same as in the previous round.
When the answers are audited, the master can safely extract the number of cheaters.
\deleted[A]{and readjust $p_\VRF$ appropriately, hence reinforcement learning is applied.}
Then, the \sout{master adapts the} auditing probability $p_\VRF$ is adapted according {\color{magenta} to the accumulated reputation}. 
Observe that the algorithm guarantees $p_\VRF \geq p_\VRF^{min}$.
This, combined with the property $p_\VRF^{min} > 0$ will
prevent the system to fall in a permanent set of ``bad" states where
$p_\VRF =0$ and $p_C >0$.
\sout{A discount factor, \added[C]{which we call {\em tolerance} and denote by $\tau$}, expresses the master's tolerable ratio of {\color{magenta} accumulate cheater reputation}}
}

We include a threshold, denoted by $\tau$, that represents the master's \emph{tolerance} to cheating
(typically, we will assume $\tau=1/2$ in our simulations). If the ratio of the aggregated reputation of cheaters with respect to the total is larger than $\tau$, $p_\VRF$ is increased, and decreased otherwise.
The amount by which $p_\VRF$ changes depends on the difference between these
values, modulated by a {\em learning rate} $\am$. This latter value determines to what extent the newly acquired information will override the old information. (For example, if $\am=0$ the master will never adjust $p_\VRF$.)
A pseudocode of the algorithm described is given as Algorithm~\ref{alg3}.

%{The algorithm guarantees 
%$p_\VRF^{min} < p_\VRF \leq 1$, $p_\VRF^{min} $ is a lower bound for the probability to guarantee that the algorithm will not fall into an infinite loop where $ p_\VRF =0$ and $p_C >0$.
%In the case workers\rq{} 
%answers are audited the master adopts it\rq{}s $p_\VRF$ according to the percentage of workers cheating. 
%A discount factor, \added[C]{which we call {\em tolerance} and denote by $\tau$}, expresses the importance that cheaters percentage has for the master. For example, if $\tau$ is small and originally the cheaters percentage is larger than $\tau$, $p_\VRF$ in the next round will be increased depending on how small $\tau$ is. The expected behavior after a series of round is for $\tau$ to be greater than the percentage of cheater and bigger $\tau$ is compared to the percentage of cheater the smaller $p_\VRF$ would be in the next round. The value $\alpha$ is the learning rate of adjusting $p_\VRF$. This value determines to what extent the newly acquired information will override the old information. For example if $\alpha=0$ will make the master to never adjust is $p_\VRF$. Thus the value of $\alpha$ describes the influence cheating workers have on $p_\VRF$.}

\paragraph{\bf \em Workers' Algorithm} \noindent
This algorithm is run only by rational workers (recall that altruistic and malicious workers have a
predefined behavior).\footnote[1]{Since the workers are not aware that a reputation scheme is used, this algorithm is the one considered in~\cite{CCPE13}; we describe it here for self-containment.} The execution of the algorithm begins with each rational worker $i$ deciding an initial probability of cheating $p_{Ci}$.  
In each round, each worker receives a task from the master and, with probability $1-p_{Ci}$ computes 
the task and replies to the master with the correct answer. Otherwise, 
it fabricates an answer, and sends the incorrect response to the master. We use a flag $S_i$ to model the stochastic decision of a worker $i$ to cheat or not. After receiving its payoff, each worker $i$ changes its $p_{Ci}$ according to payoff $\Pi_i$, the chosen strategy $S_i$, and its aspiration $a_i$. %\sout{Observe that the workers' algorithm guarantees $0\leq p_{Ci} \leq 1$. }

The workers have a {\em learning rate} $\aw$. In this work, we assume that all workers have the same learning 
rate, that is, they learn in the same manner (see also the discussion in~\cite{RLbook}; the learning rate is called step-size there); note that our analysis can be adjusted to accommodate also workers with different learning rates. 
%\textcolor{blue}{CG: I BELIEVE WE DO NOT NEED THIS ANYMORE, RIGHT? We choose the value of $\aw$ so that $\forall i %\in W: \aw(a_i+\Cp)<1$. Otherwise, the system could enter in an oscillating condition where some nodes alternate %$p_C$ between $0$ and $1$ never converging to a stable state, which is necessary to guarantee reliability.}
A pseudocode of the algorithm is given as Algorithm~\ref{alg2}.%\vspace{-0.5em} 

%Reputation-based Mechanism

%!TEX root = ./CFGMSreputation_arXiv.tex

\newcommand{\trfl}{\text{truthful }}
\newcommand{\untrfl}{\text{untruthful }}  
%\newpage
\section{Analysis}%\vspace{-.5em}}
\label{rep:analysis}

%\ec{Extend intro paragraph to better guide the reader through the section.}

We now analyze the reputation-based mechanism. We mode the evolution of the mechanism as a Markov Chain, and then discuss the necessary and sufficient conditions
for achieving eventual correctness. Modeling a reputation-based mechanism as a Markov Chain is more involved than previous models that do not consider reputation (e.g.~\cite{CCPE13}).%\vspace{-.7em}

\paragraph{\bf \em The Markov Chain} \noindent
%\label{subsec:Markov}
%We analyze the evolution of the master-workers system as a Markov chain. 
Let the state of the Markov chain be given by a vector $s$. The components of $s$ are: for the master, the probability of auditing $p_\VRF$ and the number of audits before state $s$, denoted as $aud$; and for each {\em rational} worker $i$, the probability of cheating $p_{Ci}$, the number of \emph{validations} (i.e., the worker was honest when the master audited) before state $s$, denoted as $v_i$, the error rate $\beta_i$ and $streak_i$ (the number of consecutive times a workers was found honest since the last time she cheated). 
%Malicious workers always use $p_C=1$ and altruistic workers always use $p_C=0$.  
%\sout{Generically, we refer to the number of audits and validations as \emph{counts}.} 
To refer to any component $x$ of vector $s$ we use $x(s)$. Then,
\begin{align*}
s =& \big\langle p_\VRF(s), aud(s), p_{C1}(s),p_{C2}(s),\dots,p_{Cn}(s), v_1(s), \\
& v_2(s),\dots, v_n(s), \beta_1(s), \beta_2(s),\dots, \beta_n(s), streak_1(s), \\
& streak_2(s),\dots, streak_n(s)  \big\rangle.
\end{align*} 

%\sout{Observe that the workers' reputations are not part of the Markov state. A Markov chain must be
%composed by states where the state transition depends only on the values of the previous state, not
%on a history of states. But a worker's reputation in a given round is not only due to the reputation
%of the worker in the previous round, but on a sequence of rounds. Hence, instead of including the workers' reputation in the state, we include the validations, whose state transition depends only on the previous state.
%And we can use validations to compute the reputation of a worker in a given round.} \marginpar{\miguel{I believe this observation is irrelevant}}

%\sout{As before,} 
In order to specify the transition function, we consider the execution of the protocol divided in rounds. 
%\sout{The states of the Markov chain are the values of the above specified probabilities and counts at the end of round $r$.} 
In each round, probabilities and \emph{counts} (i.e. numbers of validations, audits, error rate and streak) are updated by the mechanism as defined in Algorithms~\ref{alg3} and \ref{alg2}.
The state at the end of round $r$ is denoted as $s_r$.
Abusing the notation, we will use $x(r)$ instead of $x(s_r)$ to denote component $x$ of vector $s_r$. 
The workers' decisions, the workers' error rate, the number of cheaters, and the payoffs of each round $r>0$ are the stochastic outcome of the probabilities and counts at the end of round $r-1$.
We specify the transition from $s_{r-1}$ to $s_r$ by the actions taken by the master and the workers during round $r$.

%\sout{The master computes a reputation value $\rho_i(r)$  for each worker $i$, which is a function of  $v_i(r)$ and $a(r)$.} 

In the definition of the transition function that follows, the probabilities are limited to $p_\VRF(s) \in [p_\VRF^{min},1]$ and for each rational worker $i$ to $p_{Ci}(s) \in [0,1]$, for any state $s$.
The initial state $s_0$ is arbitrary but restricted to the same limitations. 
%\st{Also the initial state $s_0$ has $aud(s_0)=0$ and $v_i(s_0)=0$, for all $i \in W$.}
%\sout{Following the previous abbreviation,} 
%We denote the number of validations of a set of workers $X$ at the end of round $r$ as $v_{X}(r)$. 
Let $P_F(r)$ be the probability that the set of cheaters in round $r$ is exactly $F\subseteq W$. (That is, $P_F(r)=\prod_{j\in F}p_{Cj}(r-1)\prod_{k\notin F}(1-p_{Ck}(r-1))$.)
Then, the transition from state $s_{r-1}$ to $s_r$ is as follows.%\vspace{-.5em}
\begin{itemize}[leftmargin=4mm]
\item Malicious workers always have $p_C = 1$ and altruistic workers always have $p_C=0$.
\item With probability $p_\VRF(r-1) \cdot P_F(r)$, the master audits when the set of cheaters is $F$. 
Then, according to Algorithms~\ref{alg3} and \ref{alg2}, the new state is as follows.
%following updates are carried out.
\begin{itemize}
\item[] For the master: 
$aud(r) = aud(r-1) +1$ and, if $\rho_W(r)>0$ then
$p_\VRF(r) = p_\VRF(r-1) + \am \left(\rho_F(r) / \rho_W(r)-\tau\right)$ and
$p_\VRF(r) = \min\{1, p_\VRF+\am\}$ otherwise.
%\begin{itemize}
%p_\VRF(s_{r}) &= p_\VRF(s_{r-1}) + \am \left(\frac{\sum_{i \in F} (v_i(s_{r-1}) + 1)}{\sum_{i \in F} (v_i(s_{r-1}) + 1) + \sum_{i \notin F} (v_i(s_{r-1}) + 2)} - \tau\right)\\
%\item[] 
%$p_\VRF(r) = p_\VRF(r-1) + \am \left(\frac{\rho_F(r)}{\rho_W(r)}-\tau\right)$ 
%\item[] 
%$aud(r) = aud(r-1) +1$.
%\end{itemize}
\item [(1)] For each worker $i\in F$: $v_i(r) = v_i(r-1)$,  $\beta_i(r)=\beta_i(r-1)+0.1$ and $streak_i(r)=0$ and, if $i$ is rational, then $p_{Ci}(r) = p_{Ci}(r-1)  - \aw(a_i + \Cp)$.
%\begin{itemize}
%\item[]  if $i$ is rational, then $p_{Ci}(r) = p_{Ci}(r-1)  - \aw(a_i + \Cp)$
%\item[] $v_i(r) = v_i(r-1)$.
%\end{itemize}
\item [(2)] 
For each worker $i\notin F$: $v_i(r) = v_i(r-1) + 1$, $\beta_i(r)=\beta_i(r-1)\cdot 0.95$ and $streak_i(r)=streak_i(r-1)+1$ and, if $i$ is rational, then $p_{Ci}(r) = p_{Ci}(r-1) + \aw(a_i - (\SW-\Ct))$.
%\begin{itemize}
%\item[]  if $i$ is rational, then $p_{Ci}(r) = p_{Ci}(r-1) + \aw(a_i - (\SW-\Ct))$
%\item[] $v_i(r) = v_i(r-1) + 1$.
%\end{itemize}
\end{itemize}
\item With probability $(1-p_\VRF(r-1)) P_F(r)$, the master does not audit when the set of cheaters is $F$. 
Then, according to Algorithms~\ref{alg3} and \ref{alg2}, the following updates are carried out.
\begin{itemize}
\item[]
For the master: $p_\VRF(r) = p_\VRF(r-1)$ and $aud(r) = aud(r-1)$.
%\begin{itemize}
%\item[] $p_\VRF(r) = p_\VRF(r-1)$,
%\item[] $aud(r) = aud(r-1)$.
%\end{itemize}
\item[] For each worker $i\in W$: $v_i(r) = v_i(r-1)$.
%\begin{itemize}
%\item[] $v_i(r) = v_i(r-1)$.
%\end{itemize}
\item[] For each rational worker $i\in F$,
\begin{itemize}
\item [(3)] if $\rho_F(r) >\rho_{W\setminus F}(r)$ then $p_{Ci}(r) = p_{Ci}(r-1) + \aw(\SW - a_i)$, 
\item [(4)] if $\rho_F(r) <\rho_{W\setminus F}(r)$ then $p_{Ci}(r) = p_{Ci}(r-1) - \aw\cdot a_i$,
\end{itemize}
\item[] For each rational worker $i\notin F$, 
\begin{itemize}
\item [(5)] if $\rho_F(r) >\rho_{W\setminus F}(r)$ then $p_{Ci}(r) = p_{Ci}(r-1) + \aw(a_i + \Ct)$,
\item [(6)] if $\rho_F(r) <\rho_{W\setminus F}(r)$ then $p_{Ci}(r) = p_{Ci}(r-1) + \aw(a_i - (\SW-\Ct))$.%\vspace{-.5em}
\end{itemize}
\end{itemize}
\end{itemize}

%\sout{Observe that all of the above concern only rational workers. Malicious workers always have $p_C = 1$ and altruistic workers always have $p_C=0$.}
%\sout{In case of a tie in one of the transitions (3), (4), (5) or (6) the master takes a coin flip decision, choosing an answer as the majority with equal probability. }
\noindent Recall that in case of a tie in the weighted majority, the master flips a coin to choose one of the answers, and assigns payoffs accordingly. If that is the case, transitions (3)--(6) apply according 
to that outcome.%\vspace{-.7em}

%%%%%%%%%%%%%%%%%%%%%%%%%%%%%%%%%%%%%%%%%%%%%%%%%%%%%%%%%%%%%%%%%%
%%%%%%%%%%%%%%%%%%%%%%%%%%%%%%%%%%%%%%%%%%%%%%%%%%%%%%%%%%%%%%%%%%
%%%%%%%%%%%%%%%%%%%%%%%%%%%%%%%%%%%%%%%%%%%%%%%%%%%%%%%%%%%%%%%%%%

\paragraph{\bf \em Conditions for Eventual Correctness} \noindent
%\label{subsec:convergence}
We show now the conditions under which the system can guarantee eventual correctness. 
%CG: I removed the sentence below, because is used terms not defined yet.
%In order to show eventual correctness, we must show eventual convergence to a closed \trfl set. 
%%%%%%%%%%%properties stuff
%\sout{\anto{We will assume that reputation functions satisfy the following properties.}}
The analysis is carried out for a universal class of reputation functions characterized by two properties. Property~1 states that if the master audits in consecutive rounds, the aggregated reputation of the honest workers will be larger than that of cheater workers in a bounded number of rounds. Property~2 states that if the aggregated reputation of a set $X\subset W$ is larger than that of a set $Y\subset W$, then it remains so if the master audits and all workers are honest. The two properties are formally stated below.
%\vspace{-.5em}
\begin{enumerate}
%\item[]
%\label{p:limit}
%\textbf{Property 1:} For any $X\subset W$ and $Y\subset W$, if the Markov chain evolves in such a way that 
%$\forall i \in X, \lim_{r \rightarrow \infty} (v_i(r)/aud(r)) =1$ and $\forall j \in Y, \lim_{r \rightarrow \infty} (v_j(r)/aud(r)) =0$, then for any constant $\delta \leq 1$
%there is some $r^*$ such that $\forall r > r^*, \rho_X(r)>\rho_Y(r)$ and $\frac{\rho_Y(r)}{\rho_W(r)}< \delta$.
%\item \sout{if $v_i(r+1)=v_i(r)$ then  $\rho_i(r+1)\leq \rho_i(r)$ }
%\item \sout{if  $v_j(r)>v_i(r)$ then  $\rho_j(r)>\rho_i(r)$}
\item[]
\label{p:limit}
\textbf{Property 1:} 
%\ecc{
For any constant $\delta > 0$, there is a bounded value $\gamma(\delta)$ such that, for any non-empty $X\subseteq W$ and any initial state $s_r$ in which $v_i(r)=0, \forall i \notin X$, if the Markov chain evolves in such a way that $\forall k= 1,\ldots, \gamma(\delta)$, it holds that $aud(r+k)=aud(r)+k$, $\forall i \in X, v_i(r+k)=v_i(r)+k$ and $\forall j \in W \setminus X, v_j(r+k)=v_j(r)$, then $\rho_X(r+\gamma(\delta))>\delta \cdot \rho_{W \setminus X}(r+\gamma(\delta))$.\vspace{.3em}
%}
% and $\frac{\rho_Y(r)}{\rho_W(r)}< \delta$.
%\item \sout{if $v_i(r+1)=v_i(r)$ then  $\rho_i(r+1)\leq \rho_i(r)$ }
%\item \sout{if  $v_j(r)>v_i(r)$ then  $\rho_j(r)>\rho_i(r)$}
\item[]
\label{p:reputable}
\textbf{Property 2:} For any $X\subset W$ and $Y\subset W$, if  $aud(r+1)=aud(r)+1$ and $\forall j \in X \cup Y$ it is $v_j(r+1)=v_j(r)+1$ then $\rho_X(r)>\rho_Y(r) \Rightarrow \rho_X(r+1)>\rho_Y(r+1)$.\label{reputableProperty}
%\item  if $a(s_{r+1})=a(s_r)+1$, $v_i(s_{r+1})=v_i(s_r)+1$  and $v_j(s_{r+1})=v_j(s_r)$ then $\rho_i(s_{r+1})\geq\rho_i(s_r)$ and $\rho_j(s_{r+1})\leq\rho_j(s_r)$
\end{enumerate}

%\ecc{
Reputations \typeA , \typeB and \typeD (cf.~Section~\ref{rep:model})
satisfy Property~1, while reputation \typeC as defined does not. However, if a constant upper bound in the value of $\beta_i(r)$ is established, we obtain an adapted version of \typeC reputation that satisfies Property~1.
Regarding Property~2, while reputation \typeB satisfies it, reputation \typeA, \typeC and \typeD do not. 
%(The proofs of these facts can be found in~\ref{properties}.) 
The proofs of these facts are presented below. 
Moreover, as we show below (Theorem~\ref{thm}), this \emph{makes a difference} with respect 
to guaranteeing eventual correctness.

\begin{lemma}
Property 1 holds for reputation \typeA , while Property 2 does not . 
\end{lemma}

\begin{proof}
First we show that Property 1 holds. 
Consider any $d>0$, any $X \subseteq W$ non empty. Without loss of generality assume $|X|=k$. 
Consider rounds $r+1, \ldots r+j$, for some $j$, such that the master audits, workers in $X$ are honest and workers not in $X$ cheat. 
For $\forall i \in X$, $\rho_i(r+j)=\frac{v_i(r)+j+1}{aud(r)+j+2}$; and $\forall i \notin X$, $\rho_i(r+j)=\frac{v_i(r)+1}{aud(r)+j+2}$.
Then $\rho_X(r+j)=\sum_{i\in X} \rho_i(r+j)= \frac{\sum_{i \in X}v_i(r)}{aud(r)+j+2}+\frac{k(j+1)}{aud(r)+j+2}\geq\frac{j+1}{aud(r)+j+2}$; and $p_{W\setminus X}(r+j)=\sum_{i\in W \setminus X}\rho_i(r+j)=\frac{(n+k)}{aud(r)+j+2}\leq \frac{n+1}{aud(r)+j+2}$. For any $j\leq \delta(n-1)$ we have, $\rho_X(r+j)\geq \frac{j+1}{aud(r)+j+2}>\frac{\delta(n-1)}{aud(r)+j+2}\geq\rho_{W\setminus X}(r+j)$. Hence, setting $\gamma(\delta)=\delta(n-1)$ proves the first part of the claim. 
\\
\par
We now show that Property 2 does not hold. 
Consider any round where $aud(r+1)=aud(r)+1$ and $\forall j \in X \cup Y$ $v_j(r+1)=v_j(r)+1$. Without lose of generality assume that in state $s_r$, $|X|=k_r$. Then we have that if, 
$$
\mbox{\small $
\rho_X(r)>\rho_Y(r)$}$$$$
\mbox{\small $
\sum_{i\in X}\frac{v_i(r)+1}{aud(r)+2}>\sum_{i\in Y}\frac{v_i(r)+1}{aud(r)+2}$}$$$$
\mbox{\small $
\frac{\sum_{i\in X}v_i(r)}{aud(r)+2}+\frac{k_r}{aud(r)+2}>\frac{\sum_{i\in Y}v_i(r)}{aud(r)+2}+\frac{n-k_r}{aud(r)+2}$}$$$$
\mbox{\small $
\sum_{i\in X}v_i(r)+k_r>\sum_{i\in Y}v_i(r)+n-k_r
$}
$$
then,
$$
\mbox{\small $
\rho_X(r+1)>\rho_Y(r+1)$}$$$$
\mbox{\small $
\sum_{i\in X}\frac{v_i(r)+2}{aud(r)+3}>\sum_{i\in Y}\frac{v_i(r)+2}{aud(r)+3}$}$$$$
\mbox{\small $
\frac{\sum_{i\in X}v_i(r)}{aud(r)+3}+\frac{2k_{r+1}}{aud(r)+3}>\frac{\sum_{i\in Y}v_i(r)}{aud(r)+3}+\frac{2(n-k_{r+1})}{aud(r)+3}$}$$$$
\mbox{\small $
\sum_{i\in X}v_i(r)+2k_{r+1}>\sum_{i\in Y}v_i(r)+2(n-k_{r+1})
$}
$$

Thus if $k_r \neq k_{r+1}$ then the entailment may not hold. 
\end{proof}

\begin{lemma}
Property 1 and 2 hold for reputation \typeB . 
\end{lemma}

\begin{proof}First we show that Property 1 holds. 
Consider any $\delta>0$, any $X \subseteq W$ non empty. Without loss of generality assume $|X|=k$. Consider rounds $r+1, \ldots r+j$, for some $j$, such that the master audits, workers in $X$ are honest and workers not in $X$ cheat. 
For $\forall i \in X$, $\rho_i(r+j)= \varepsilon^{aud(r)-v_i(r)}$ and $\forall i \notin X$, $\rho_i(r+j)=\varepsilon^{aud(r)+j-v_i(r)}$.Then $\rho_X(r+j)= \sum_{i\in X}\rho_i(r+j)= \sum_{i\in X}\varepsilon^{aud(r)-v_i(r)}\geq \varepsilon^{aud(r)}$ and $\rho_{W\setminus X}(r+j)=\sum_{i\in W\setminus X} \rho_{i}(r+j)=\sum_{i\in W\setminus X} \varepsilon^{aud(r)+j-v_i(r)}\leq (n-1)\varepsilon^{aud(r)+j}$. For any $j< -log(\delta(n-1))$ we have, 
$\rho_X(r+j)\geq \varepsilon^{aud(r)} > \delta (n-1)\varepsilon^{aud(r)+j} \geq \rho_{W\setminus X}(r+j)$. Hence, setting $\gamma(\delta)<-log(\delta(n-1))$ proves the claim. 
\\
\par
Now we show that also Property 2 holds. 
Consider any $X \subset W$ and $Y \subset W$. Then if $\rho_X(r)= \sum_{i\in X} \varepsilon^{aud(r)-v_i(r)}$ then $\rho_X(r+1)= \sum_{i\in X} \varepsilon^{aud(r)+1-v_i(r)-1}= \sum_{i\in X} \varepsilon^{aud(r)-v_i(r)}$. If $\rho_Y(r)= \sum_{i\in Y} \varepsilon^{aud(r)-v_i(r)}$ then $\rho_Y(r+1)= \sum_{i\in Y} \varepsilon^{aud(r)+1-v_i(r)-1}= \sum_{i\in Y}\varepsilon^{aud(r)-v_i(r)}$. Thus trivially the condition $\rho_X(r)>\rho_Y(r)\Rightarrow	\rho_X(r+1)>\rho_Y(r+1)$ holds. 
\end{proof}

\begin{lemma}
Property 1 does not hold for reputation \typeC, unless an upper bound $b$ of the value $\beta_i(r)$ is established;Property 2 does not hold for reputation \typeC . 
\end{lemma}

\begin{proof} First we show that Property 1 does not hold. 
Consider any $\delta>0$, and $X \subseteq W$ non empty. Without loss of generality assume $|X|=k$. Consider rounds $r+1 \ldots r+j$ for some $j$, such that master audits, workers in $X$ are honest and workers not in $X$ cheat. 
For $\forall i \in X$ if $\beta_i(r+j)>0.05$ then $\rho_i(r+j)=0$ else $\rho_i(r+j)=1-\sqrt{\frac{\beta_i(r)\times j\times 0.95}{0.05}}$. For $\forall i \in W \setminus X$ then $\beta_i(r+j)>0.05$ thus $\rho_i(r+j)=0$. If $\gamma(\delta)>0$ then $\forall i \in W \setminus X$ $\rho_i(r+j)=0$.To have $\rho_X(r+\gamma(\delta))> \delta \rho_{W \setminus X}(r+\gamma(\delta))$ we need to know the state $s_r$ where the master starts to audit to know in how many rounds $\exists i \in X$ where $\beta_i(r)\leq 0.05$. Thus the condition of  Property 1 for \typeC depends on the current state $s_r$ where the master begins to audit. Hence Property 1 for reputation \typeC does not hold.  
\\
\par
Now, we show that if an upper bound $b$ of the value $\beta_i(r)$ is established then Property 1 holds for reputation \typeC. 
Consider any $\delta>0$, and $X \subseteq W$ non empty. Without loss of generality assume $|X|=k$. Consider rounds $r+1 \ldots r+j$ for some $j$, such that master audits, workers in $X$ are honest and workers not in $X$ cheat. 
For $\forall i \in X$ if $\beta_i(r+j)>0.05$ then $\rho_i(r+j)=0$ else $\rho_i(r+j)=1-\sqrt{\frac{\beta_i(r)\times j\times 0.95}{0.05}}$. For $\forall i \in W \setminus X$ then $\beta_i(r+j)>0.05$ thus $\rho_i(r+j)=0$. If $\gamma(\delta)>0$ then $\forall i \in W \setminus X$ $\rho_i(r+j)=0$. 
If a constant upper bound $b$ in the value of $\beta_i(r)$ is established in the $s_r$ state then in $j<\frac{0.05}{b\times 0.95}$ rounds $\forall i \in X$, $\rho_i(r+j)>0$ and thus the condition $\rho_X(r+\gamma(\delta))> \delta \rho_{W \setminus X}(r+\gamma(\delta))$ becomes true by setting $\gamma(\delta)<\frac{0.05}{b \times 0.95}$ and the claim is proved. 
\\
\par
Finally, we show that Property 2 does not hold. 
Consider any $X \subset W$ and $Y \subset W$. If $\rho_X(r)>\rho_Y(r)$ then in the next state $s_{r+1}$ the following apply. For $\forall j \in X \setminus Y$, $v_j(r+1)=v_j(r)+1$ (they are honest). Thus if $\beta_j(r+1)>0.05$ then $\rho_j(r+1)=0$ else $\rho_j(r+1)=1-\sqrt{\frac{\beta_j(r)\times 0.95}{0.05}}$. Consider the following case where $\forall j \in Y$, $\rho_j(r)=0$ and $\forall j \in X'$, where $X'$ is the set of all workers in $X$ besides one, lets name it $z$, $\rho_j(r)=0$ and $\rho_z(r)=1-\sqrt{\frac{\beta_z(r)\times 0.95}{0.05}}$. Then $\rho_X(r)>\rho_Y(r)$ holds. In the next round though there is a possibility that $\forall j \in Y$, $\rho_j(r+1)=1-\sqrt{\frac{\beta_j(r+1)\times 0.95}{0.05}}$, while only the reputation of $z$ changes in the next state for the $X$ set. Thus if $ 1-\sqrt{\frac{\beta_z(r)\times 0.95}{0.05}}<|Y|\times(1-\sqrt{\frac{\beta_j(r+1)\times 0.95}{0.05}})$ then Property 2 does not hold for reputation \typeC and the claim is proved.     
\end{proof}

\begin{lemma}
Property 1 holds for reputation \typeD, while Property 2 does not hold when $n>2$.
\end{lemma}

\begin{proof}
First we show that Property 1 holds. 
Consider any $\delta>0$, and $X \subseteq W$ non empty. Without loss of generality assume $|X|=k$. Consider rounds $r+1 \ldots r+j$ for some $j$, such that master audits, workers in $X$ are honest and workers not in $X$ cheat. 
For $\forall i \in X$  we have $\rho_i(r+j)=1-\frac{1}{streak_i(r+j)}$ if $streak_i(r+j)\geq 10$. For $\forall  i \notin X$ $\rho_i(r+j)=0$ since $streak_i(r+j)=0$. Thus, $\rho_X(r+j)=\sum_{i\in X}\rho_i(r+j)=\sum_{i\in X}(1-\frac{1}{streak_i(r+j)})=k(1-\frac{1}{streak_i(r+j)})>\rho_{W \setminus X}(r+j)=\sum_{i \in W \setminus X}\rho_i(r+j)=0$. 
Thus, in order for $\rho_X(r+\gamma(\delta))>\delta\rho_{W\setminus X(r+\gamma(\delta))}$ to be true we need to set $\gamma(\delta)\geq 10$.
\\
\par
We now show that Property 2 does not hold when $n>2$.
Consider any round where $aud(r+1)=aud(r)+1$ and $\forall j \in X \cup Y$ $v_j(r+1)=v_j(r)+1$. 
Then, for any $X \subset W$ and $Y \subset W$ it must hold that if $\rho_X(r)>\rho_Y(r)$ then $\rho_X(r+1)>\rho_Y(r+1)$ .
Thus, consider state $s_r$ where $|X|=1$, $|Y|=n-1$ and $\forall i \in X$ $streak_i(r) \geq 10$ and $\forall i \in Y$ $streak_i(r)=9$. 
It is true that $\rho_X(r)=1-\frac{1}{streak_i(r)}=1-\frac{1}{10+\delta}$ $>\rho_Y(r)=(n-1)\cdot 0=0$ where $\delta>1$. 
Now, in the next state $s_{r+1}$ we have $\rho_X(r+1)=1-\frac{1}{streak_i(r+1)}=1-\frac{1}{10+\delta+1}<1$ and $\rho_Y(r+1)=(n-1)(1-\frac{1}{10})$ $ 	\Rightarrow  n-2 <\rho_Y(r+1)<n-1$. If $n>2$ then $\rho_X(r+1)<\rho_Y(r+1)$ and the claim does not. Thus Property 2 does not hold for reputation \typeD when $n>2$.   
\end{proof}

Moving on we present the conditions under which the system can
guarantee eventual correctness, but before that we establish the 
terminology that will be used throughout. For any given state $s$, a set $X$ of workers is called a \emph{reputable set} if 
$\rho_X(r)>\rho_{W\setminus X} (r)$.
%\marginpar{\miguel{I think it should be $v_X(s)>v_{W\setminus X} (s)$.}}
%\sout{Let a worker $i$ be called a \emph{covered worker} if it is paid more than its aspiration $a_i$ plus the computing cost $\Ct$ i.e. $ \SW > a_i +\Ct$. }\marginpar{\miguel{moved down}}
In any given state $s$, let a worker $i$ be called an \emph{honest worker} if $p_{Ci}(s)=0$.
%\sout{Let a state $s$ be called a \emph{\trfl state} if $H$ is the set of honest workers in state $s$ and $H$ is reputable.} 
Let a state $s$ be called a \emph{\trfl state} if the set of honest workers in state $s$ is reputable.
Let a \emph{\trfl set} be any set of \trfl states. 
Let a worker be called a \emph{covered worker} if the payoff of returning the correct answer is at least its aspiration plus the computing cost. I.e., for a covered worker $i$, it is $\SW \geq a_i +\Ct$. 
We refer to the opposite cases as \emph{uncovered worker} ($\SW < a_i +\Ct$), \emph{cheater worker} 
($p_{Ci}(s)=1$), \emph{\untrfl state} 
%\sout{($F$ is the set of workers with $p_{Ci}(s)=1$ in state $s$ and $F$ is reputable)} 
(the set of cheaters in that state is reputable), and \emph{\untrfl set}, respectively.
Let a set of states $S$ be called \emph{closed} if, once the chain is in any state $s\in S$, it will not move to any state $s'\notin S$. (A singleton closed set is called an \emph{absorbing} state.) For any given set of states $S$, we say that the chain \emph{reaches} (resp. \emph{leaves}) the set $S$ if the chain reaches some state $s\in S$ (resp. reaches some state $s\notin S$). 

%\sout{In order to show eventual correctness, we must show eventual convergence to a closed \trfl set. We first consider that {\em all workers are rational}, and then we discuss (last paragraph) the case where we also have malicious and altruistic workers.} \marginpar{\miguel{moved up}}
%\sout{As in the analysis of the basic mechanism,} 

In the master's algorithm, a non-zero probability of auditing is always guaranteed. This is a necessary condition. Otherwise, unless the altruistic workers outnumber the rest, a closed untruthful set is reachable, as we show in Lemma~\ref{pA0b_ap}. %in Appendix~\ref{appendix_analysis}.}

\begin{lemma} 
\label{pA0b_ap}
Consider any set of workers $Z\subseteq W$ such that $\SW> a_i$, for every rational worker $i\in Z$.
Consider the set of states $$S=\{s | (p_\VRF(s)=0) \land (\forall w\in Z : p_{Cw}(s)=1) $$ $$ \land 
(\rho_Z(s) > \rho_{W-Z}(s)) \}.$$ Then,
\begin{itemize} 
\item [\emph{(i)}] $S$ is a closed untruthful set, and 
\item [\emph{(ii)}] if
$p_\VRF(0)=0$,
$\rho_Z(0) > \rho_{W-Z}(0)$, and
for all $i\in Z$ it is $p_{Ci}(0)>0$,  
then, $S$ is reachable.
\end{itemize}
\end{lemma}

\begin{proof}

\emph{(i)} Each state in $S$ is untruthful, since the workers in $Z$ are all cheaters and $Z$ is a reputable set. 
Since $p_\VRF=0$, the master never audits, and the reputations are never updated. 
From transition (3) it can be seen that, if the chain is in a state 
of the set $S$ before round $r$, for each worker $i \in Z$, it holds $p_{Ci}(r) \geq p_{Ci}(r-1) =1$.  
Hence, once the chain has reached a state in the set $S$, it will move only to states in the set $S$. Thus, $S$ is a closed untruthful set.

(\emph{ii}) We show now that $S$ is reachable from the initial state under the above conditions. Because $p_\VRF$ and the reputations only change when the master audits, we have that $p_\VRF(0)=0  \Longrightarrow p_\VRF(s)=0$ and $\rho_Z(0) > \rho_{W-Z}(0) \Longrightarrow \rho_Z(s) > \rho_{W-Z}(s)$, for any state $s$. 
Malicious workers always have $p_C=1$, and no altruistic worker may be contained in $Z$ because $p_{Ci}(0)>0$ for all $i\in Z$. Thus, to complete the proof it is enough to show that eventually it is $p_C=1$ for all the workers in L, which is the set of rational workers in $Z$.
Given that for each rational worker $j\in L$, $p_{Cj}(0)>0$ and $\SW> a_j$, from transition (3) it can be seen that there is a non-zero probability of moving from $s_0$ to a state $s_1$ where the same conditions apply and $p_{Cj}(1)>p_{Cj}(0)$ for each rational worker $j\in L$. Hence, applying the argument inductively, there is a non-zero probability of reaching $S$.
\end{proof}

%From these definitions, it follows that e
Eventual correctness follows if we can show that the Markov chain always ends in a closed \trfl set, with $p_\VRF=p_\VRF^{min}$.
We prove first that having at least one worker that is altruistic or covered rational is necessary for a closed \trfl set to exist. Then we prove that it is also sufficient if all rational workers are covered.%\vspace{-.5em}

\begin{lemma}
\label{l:nocovered}
If all workers are malicious or uncovered rationals, no \trfl set $S$ is closed, if the reputation type satisfies Property~2.%\vspace{-.5em}
\end{lemma}

\begin{proof}
Let us consider some state $s$ of a \trfl set $S$. Let $Z$ be the set of honest workers in $s$. Since $s$ is \trfl,
then $Z$ is reputable.
Since there are no altruistic workers, the workers in $Z$ must be uncovered rational. Let us assume that being in state $s$ the master audits in round $r$.
From Property~2, since all nodes in $Z$ are honest in $r$, $Z$ is reputable after $r$. From transition (2), after round $r$, each worker $i \in Z$ has
$p_{Ci}(r) >0$. Hence, the new state is not \trfl, and $S$ is not closed. 
\end{proof}

\begin{lemma}
\label{l:onecovered}
%\ecc{
Consider a reputation type that satisfies Properties~1 and 2.
If all rational workers are covered and at least one worker is altruistic or rational, a closed \trfl set $S$ is reachable from any initial state. Moreover, in every state $s\in S$, $p_\VRF(s) = p_\VRF^{min}$.
%}
%\vspace{-.5em}
\end{lemma}

\begin{proof}
%\ecc{MODIFIED THE WHOLE PROOF}
Let $X$ be the set of altruistic and rational workers, and consider any initial state $s_r$. 
Let us define a constant $\delta=\max\{1, (1-\tau+\eta/\am)/(\tau+\eta/\am)\}$, for a fixed constant $\eta \in (0,\tau \am)$.
We consider the following cases.
\begin{enumerate}
\item
In state $s_r$ not all the workers in $X$ are truthful. %\st{(i.e., some rational worker $j \in X$ has $p_{Cj}>0$)}.
Let us assume then that in the next $\lceil \frac{1}{\aw(a_j - \Cp)} \rceil$ rounds the master audits and any
worker $i$ that has $p_{Ci}>0$ in the round cheats. Then, from transition (1) and the fact that all rational workers are covered, 
after these $\lceil \frac{1}{\aw(a_j - \Cp)} \rceil$ rounds all the workers in $X$ are truthful. Then, we end up in one of the 
following three cases.\vspace{.3em}
\item
In state $s_r$ all the workers in $X$ are truthful, and % \st{(i.e., $\forall j \in X, p_{Cj}=0$)}, and 
$\rho_X(r) \leq \delta \cdot \rho_{W \setminus X}(r)$.
Consider the value $\gamma(\delta)$ given in Property 1. 
%Let us assume that in the next $\gamma(\delta)$ rounds the master audits \ecc{and $v_i(r)=0$ $\forall i \notin X$}.
Assume that in each of the following $\gamma(\delta)$ rounds the master audits. The workers in $W\setminus X$ are malicious, hence, it holds that $\forall \notin X:v_i(r)=0$.
Then, in these rounds all workers in $X$ are honest (every worker in $X$ remains truthful from transition (2) and the fact that all rational workers are covered) and 
%all workers in $W\setminus X$ cheat \ecc{and $v_i(r)=0$ $\forall i \in W\setminus X$} (they are all malicious). 
all workers in $W\setminus X$ cheat because they are malicious. Therefore, it holds that $\forall i\notin X: \forall j \in [r,r+\gamma(\delta)]: v_i(j)=0$.
Then, from Property 1, after the $\gamma(\delta)$ rounds we have
that $\rho_X(r+\gamma(\delta))>\delta \cdot \rho_{W \setminus X}(r+\gamma(\delta))$. Then, we are in one of the following two cases.\vspace{.3em}
\item 
In state $s_r$ all the workers in $X$ are truthful, %\st{(i.e., $p_{Cj}=0$)}, 
$\rho_X(r) > \delta \cdot \rho_{W \setminus X}(r)$, and $p_\VRF(r) > p_\VRF^{min}$.
Let us assume that in the next $\lceil p_\VRF(r)/\eta \rceil$ rounds the master audits. Then, as in the previous case, in these 
rounds all workers in $X$ are honest 
%\st{(every worker in $X$ remains truthful from transition (2) and the fact that all rational workers are covered)} 
and all workers in $W\setminus X$ cheat. %\st{(they are all malicious)}. 
Hence, the property that $\rho_X(r+k) > \delta \cdot \rho_{W \setminus X}(r+k)$ holds for each round $r+k$, for $k=1, \ldots, \lceil p_\VRF(r)/\eta \rceil$. Then, by the definition of $\delta$ and the update of $p_\VRF$, in each round $p_\VRF$ is decremented by
$\eta$ (more precisely, by $\min\{\eta,p_\VRF\}$). Hence, by round $r+\lceil p_\VRF(r)/\eta \rceil$ it holds that $p_\VRF = p_\VRF^{min}$. Then, we are in the following case.\vspace{.3em}
\item 
In state $s_r$ all the workers in $X$ are truthful, %\st{(i.e., $p_{Cj}=0$)}, 
$\rho_X(r) > \delta \cdot \rho_{W \setminus X}(r)$, and $p_\VRF(r) = p_\VRF^{min}$.
Then, all subsequent states satisfy all these properties, and define the set $S$, independently of whether
the master audits or not (from transition (2) and (6), the fact that $\delta \geq 1$, Property~2, and the update of $p_\VRF$).
This complete the proof.
\end{enumerate}
\end{proof}

Now, combining Lemmas~\ref{l:nocovered}~and~\ref{l:onecovered} we obtain the following theorem. 

\begin{theorem}
\label{thm}
%\ecc{
In a system where (1) the reputation type used satisfies Properties~1 and 2, and (2) all rational workers are covered, having at least one altruist or rational worker is a necessary and sufficient condition to guarantee eventual correctness. That is, from any initial state, to eventually reach a closed \trfl set $S$ where the master audits with probability  $p_\VRF^{min}$.
%\vspace{-.5em}
\end{theorem}

%\vspace{5mm}
If there is no knowledge on the distribution of the workers among the three types (altruistic, malicious and rationals), the only strategy to make sure eventual correctness is achieved, if possible, is to cover all workers. Of course, if {\em all} workers are malicious there is no possibility of reaching eventual correctness.%\vspace{-1.2em}

 %markov analysis

%!TEX root =  access.tex

\newcommand{\myboldmath}{}%\boldmath doesn't work with latex
\newcommand{\defn}[1]           {{\textit{\textbf{\myboldmath #1}}}}
%\newcommand{\defi}[1]           {{\textit{\textbf{\myboldmath #1\/}}}}
%\newpage
\section{Simulations}

Our analysis has shown that reaching eventual correctness is feasible under certain conditions. Once the system enters a state of eventual correctness we are in an optimal state where the master always receives the correct task reply by auditing with a minimum probability. 
What is left to be clarified is under which cost eventual correctness is reached. Cost can be measured in terms of 1) reliability, 2) auditing, 3) payment to the workers and  4) time until eventual correctness is reached. Under these parameters we provide a comparison of the system's performance under the different reputation types and we are able to identify the scenarios under which every reputation type is performing best. 

We present simulations for a variety of parameter combinations similar to the values observed in real systems (extracted from~\cite{einstein,emBoinc}). We have designed our
own simulation setup by implementing our mechanism (the master\rq{}s and the workers\rq{} algorithms, including the four types
of reputation discussed above) using C++. The simulations were contacted on a dual-core AMD Opteron 2.5GHz processor,
with 2GB RAM running CentOS version 5.3.

\subsection*{\bf\em General setting}
\noindent
We consider a total of 9 workers as an appropriate degree of redundancy to depict the changes that different ratios of rational, altruistic or malicious workers will induce in the system. SETI-like systems usually use three workers, but using such a degree of redundancy would not allow us to present a rich account of the system's evolution. Additionally, by selecting 9 redundant workers we are able to capture systems that are more critical and aim at a higher degree of redundancy. The chosen parameters are indicated in the figures. As for the intrinsic parameter of the aspiration level we consider for simplicity of presentation that all workers have the same aspiration level $a_i= 0.1$; although we have checked that with random values the results are similar to those presented here, provided their variance is not very large. We set the learning rate to a small constant value, as it is discussed in~\cite{RLbook} (called step-size there), this is the general conversion when a learning process is assumed. Thus we consider the same learning rate for the master and the workers, i.e., $\alpha = \am = \aw=0.1$. We set $\tau=0.5$ (which means that the master will not tolerate a majority of cheaters), $p_\VRF^{min}=0.01$ and $\varepsilon=0.5$ in reputation \typeB.
We use $\SW=1$, as the normalization parameter for all the results presented. 
Finally, the presented results are an average of 10 executions of the implementation, unless otherwise stated (when we show the behavior of typical, individual realizations). Usually we choose to depict the evolution of $p_\VRF$ since it is an important measure of cost for the master. 
In all of the depicted results in all the Figures presented here we have verified that if $p_\VRF=p_\VRF^{min}$ then the system has already reached a state where the master receives always the correct reply, and hence eventual correctness is reached. By convention for clarity of presentation, we will simply say that the system has reached convergence once $p_\VRF=p_\VRF^{min}$. Finally, we define ${\sum_{i\in W} \rho_i S_i}/{|W|}$ as the reputation ratio. This quantity will allow us to see the overall reputation of the workers in the system and is indicative of the existence of honest workers with higher overall reputation than cheaters.

%The figures represent the average over 10 executions of the implementation, unless otherwise stated (when we show the behavior of typical, individual realizations). The chosen parameters are indicated in the figures. 
%For simplicity,  we consider that all workers have the same aspiration level $a_i= 0.1$, although we have checked that with random values the results are similar to those presented here, provided their variance is not very large.
%We consider the
%same learning rate for the master and the workers, i.e., $\alpha = \am = \aw=0.1$.
%Note that the
%learning rate, as discussed for example in~\cite{RLbook} (called step-size there), is generally set to a small constant value for practical reasons. 

%\evgenia{Finally, we set $\tau=0.5$ (which means that the master will not tolerate a majority of cheaters), $p_\VRF^{min}=0.01$ and $\varepsilon=0.5$ in reputation \typeB.}

%\begin{figure*}[ht]
%\begin{center}
%$\begin{array}{cc}
%\includegraphics[width=2.7in, trim = 1.1mm 0mm 2mm 2mm, clip]{combine_pA_no_malicious_pc05-xi05-wpc0-wct01-alpha01-asp01}&
%\includegraphics[width=2.7in, trim = 1.1mm 0mm 2mm 2mm, clip]{combine_pA_no_malicious_pc1-xi05-wpc0-wct01-alpha01-asp01}\\
%~~~~~~~~~~(a)&~~~~~~~~~~(b)
%\end{array}$
%\end{center}
%\caption{Rational workers. Auditing probability of the master as a function of time (number of rounds) for parameters $p_\VRF=0.5$, $\alpha=0.1$, $a_i=0.1$, $\tau=0.5$, $\SW=1$, $\Cp=0$ and $\Ct=0.1$. (a) initial $p_C=0.5$ (b) initial $p_C=1$. \vspace{-.5em}}
%\label{pA-no-malicious}
%\end{figure*}
\begin{figure}[htbp]
\centering
\subfloat[][]{\includegraphics[width=.42\textwidth]{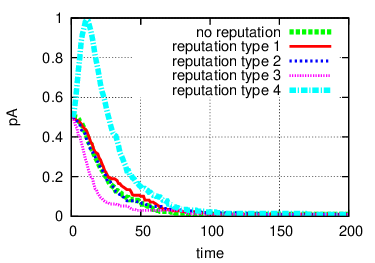}}
\qquad
\subfloat[][]{\includegraphics[width=.42\textwidth]{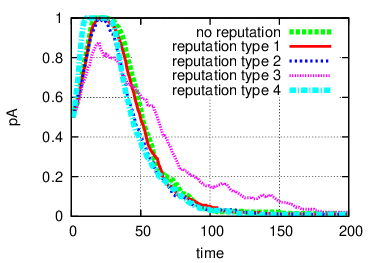}}
\caption{Rational workers. Auditing probability of the master as a function of time (number of rounds) for parameters $p_\VRF=0.5$, $\alpha=0.1$, $a_i=0.1$, $\tau=0.5$, $\SW=1$, $\Cp=0$ and $\Ct=0.1$. (a) initial $p_C=0.5$ (b) initial $p_C=1$.}
\label{pA-no-malicious}
\end{figure}

%\begin{figure*}[ht]
%%\hspace*{-0.8in}
%\begin{tabular}{cc}
%%\hspace{2.5em}
%\includegraphics[width=2.7in, trim = 1mm 0mm 2.5mm 2mm, clip]{plot-rep-pc05-xi05-wpc0-wct01-alpha01-asp01-rep_t1_pc05}&
%\includegraphics[width=2.7in, trim = 1mm 0mm 2.5mm 2mm, clip]{plot-rep-pc05-xi05-wpc0-wct01-alpha01-asp01-rep_t2_pc05}\\
%~~~~~~~~(a)&~~~~~~~~(b)\\
%\includegraphics[width=2.7in, trim = 1mm 0mm 2.5mm 2mm, clip]{plot-rep-pc05-xi05-wpc0-wct01-alpha01-asp01-rep_t3_pc05}&
%\includegraphics[width=2.7in, trim = 1mm 0mm 2.5mm 2mm, clip]{plot-rep-pc05-xi05-wpc0-wct01-alpha01-asp01-rep_t4_pc05}\\
%~~~~~~~~(c)&~~~~~~~~(d)
%\end{tabular}
%\caption{Rational workers, for an individual realization with initially $p_C=0.5$, $p_\VRF=0.5$, $\tau=0.5$, $\SW=1$, $\Ct=0.1$, $\Cp=0$, $\alpha=0.1$  and $a_i=0.1$. (a) reputation \typeA , (b) reputation \typeB , (c) reputation \typeC , (d) reputation \typeD . %(Simulations for  $p_C=1$ can be found in Figure~\ref{rep-no-malicious_appendix} in Appendix~\ref{appendix_simulations} )
%\vspace{-.5em}}
%\label{rep-no-malicious}
%\end{figure*}
\begin{figure*}[htbp]
 \begin{minipage}{\textwidth}
\centering
\subfloat[][]{\includegraphics[width=.42\textwidth]{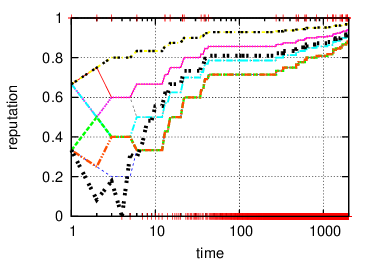}}
\subfloat[][]{\includegraphics[width=.42\textwidth]{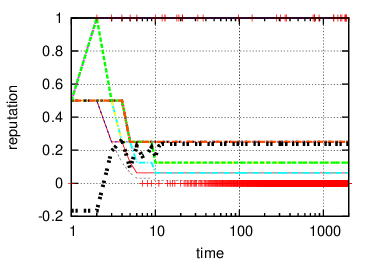}}
\end{minipage}
 \begin{minipage}{\textwidth}
\centering
\subfloat[][]{\includegraphics[width=.42\textwidth]{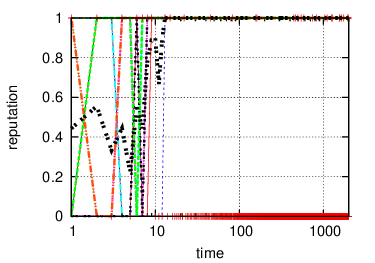}}
\subfloat[][]{\includegraphics[width=.42\textwidth]{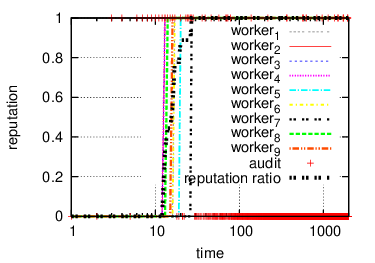}}
\end{minipage}
\caption{Rational workers, for an individual realization with initially $p_C=0.5$, $p_\VRF=0.5$, $\tau=0.5$, $\SW=1$, $\Ct=0.1$, $\Cp=0$, $\alpha=0.1$  and $a_i=0.1$. (a) reputation \typeA , (b) reputation \typeB , (c) reputation \typeC , (d) reputation \typeD .}
\label{rep-no-malicious}
\end{figure*}

%The contents of this section can be summarized as follows: In the next paragraph  we present results considering only rational workers and, subsequently, results involving all
%three type of workers. We continue with a discussion on the number of workers that if covered, eventual correctness is reached, and how the choice of
%reputation affects this. Finally, we show 1that our mechanism is robust even in the event of having workers changing
%their behavior (e.g., rational workers becoming malicious due to software or hardware errors). Here we present a rich account of the most representative figures; other supportive figures can be found in~\cite{TR}.
   
We consider a variety of different scenarios where only rational workers exist in the computation or where the master decided to cover the aspired amount of payment to only a small number of workers. We also consider the case where more than one type of workers co-exist in the same computation. Finally we also consider the case where some workers change type after eventual correctness is reached. Under these rich account of difference scenarios we are able to compare the four reputation types and record the system' s behavior before eventual correctness is reached. Notice that in all the figures' legends we refer to reputation type \typeA as type~1, to type \typeB as type~2, to type \typeC as type~3 and finally to type \typeD as type~4. 

\subsection*{\bf\em Presence of only rational workers} In crowdsourcing systems like Amazon's Mechanical Turk the majority of workers participating in the platform are driven by monetary incentives, thus exhibiting a rational behavior where their goal is to maximize their profit. Hence, the presence of only rational workers is plausible in some real system examples. In this scenario we cover all the workers, that is, $\SW>a+\Ct$. 

Figure~\ref{pA-no-malicious} depicts the auditing probability of the master at each round for all 4 reputation types and the case where a mechanism without reputation is used (see~\cite{CCPE13}). Figure~\ref{pA-no-malicious} (a) shows the case where the rational workers linger between cheating or being honest in the first round of interaction by setting $p_C=0.5$. Also the master takes an approach of ignorance by setting $p_\VRF=0.5$ and not punishing the workers. Under this mild approach of the master in all 4 reputation types the system converges in roughly 100 rounds. While in the case where the master does not use reputation the system converges a bit earlier. On the other hand $p_\VRF$ at each round is the lowest for reputation \typeC while it is the highest for reputation \typeD . In particular for \typeD the $p_\VRF$ increases in the initial rounds before decreasing. This behavior is correlated with the fact that \typeD is a type that is based on a threshold, it needs 10 consecutive correct replies for a worker to increase her reputation from zero. This is also verified by the evolution of the reputation of \typeD in Figure~\ref{rep-no-malicious} (d). 
Reputation \typeC (Figure~\ref{rep-no-malicious}(c)), on the other hand, allows for dramatic increases and decreases of reputation. This is a result of the indirect way we calculate reputation \typeC , as we mentioned above.
Notice that in reputation \typeB (Figure~\ref{rep-no-malicious}(b)) reputation takes values between (0,0.3). This happens because  
when the master catches a worker cheating, its reputation decreases exponentially, never increasing again. 
Finally, reputation \typeA (Figure~\ref{rep-no-malicious}(a)) leads rational workers to reputation values close to 1 {(at a rate that depends on the value of the initial $p_C$)} since it is a linear function.

In Figure~\ref{pA-no-malicious} (b) now we assume that the workers are more aggressive towards the system starting with an initial  $p_C=1$ . In this case convergence comes roughly around 120-150 rounds for reputation \typeA , \typeB , \typeD and  the case of no reputation. For reputation 	\typeC convergence comes later, roughly in 200 rounds, but $p_\VRF$ until convergence is lower in the first 50 rounds.  

In general from Figure~\ref{pA-no-malicious} we can see that  the mechanism of~\cite{CCPE13} (without the reputation scheme) is enough to bring rational workers to produce the correct output, precisely because of their rationality. Thus if the master can be certain that only rational workers will take part in the computation is better to use the mechanism of~\cite{CCPE13}. But if such a knowledge is not available then selecting reputation \typeB is the best option. Reputation \typeC performs better when $p_C=0.5$ while it has a poor performance in the case where $p_C=1$. As for reputation \typeA it is always slightly under performing \typeB and reputation \typeD has a bad performance when $p_C=0.5$.

\begin{figure}[htbp]
\centering
\includegraphics[width=2.8in, trim = 1.1mm 0mm 2mm 2mm, clip]{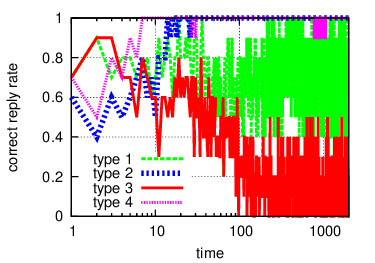}
\caption{Correct reply rate as a function of time in the presence of only rational workers. Parameters are initial $p_\VRF=0.5$,$\SW=1$, $\Ct=0.1$, $\Cp=0$ and $\alpha=0.1$, $a_i=0.1$, $\tau=0.5$.  Reputation types \typeA, \typeC and \typeD have initial $p_C=0.5$, while in \typeB , $p_C=1$.}
\label{fig4A}
\end{figure}

\subsection*{\bf\em Covering only a subset of  rational  workers} \noindent

In the previous paragraph we considered only cases where the master was covering all workers, that is, $\SW>a+\Ct$ for all workers. For the case with malicious workers, as explained in Section~\ref{rep:analysis}, this is unavoidable if the worker's type distribution is not known. 
But if we know that only rational workers exist then maybe by covering only a subset of them the system can reach eventual correctness, a scenario that we now explore. Covering only a subset of the rational workers will decrease the cost of the master in terms of payment but might actually increase the cost of auditing. This precisely is the relationship we want to explore.  
%
%\begin{figure}[htbp]
%\begin{center}$
%\begin{array}{cc}
%\includegraphics[width=2.8in, trim = 1.1mm 0mm 2mm 2mm, clip]{log-combine-1-covered-xi05-wpc0-wct01-alpha01-asp01}
%\end{array}$
%\end{center}
%\caption{Correct reply rate as a function of time in the presence of only rational workers. Parameters are initial $p_\VRF=0.5$,$\SW=1$, $\Ct=0.1$, $\Cp=0$ and $\alpha=0.1$, $a_i=0.1$, $\tau=0.5$.  Reputation types \typeA, \typeC and \typeD have initial $p_C=0.5$, while in \typeB , $p_C=1$.}
%\label{fig4A}
%\end{figure}

\begin{figure*}[htbp]
%\hspace*{-1in}
\begin{minipage}{\textwidth}
$\begin{array}{ccc}
\hspace{4em}
\includegraphics[width=2in, trim = 1.2mm 0mm 0mm 2mm, clip]{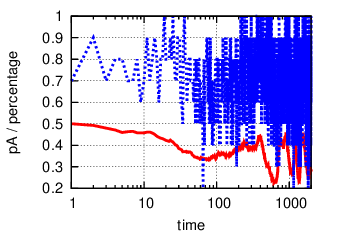}&
\includegraphics[width=2in, trim = 1.2mm 0mm 0mm 2mm, clip]{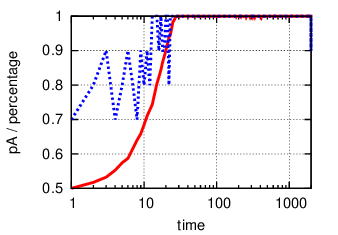}&
\includegraphics[width=2in, trim = 1.2mm 0mm 0mm 2mm, clip]{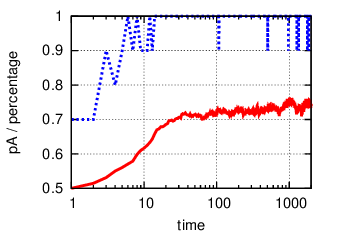}\\
~~~~~~~~(a1)&~~~~~~~~(b1)&~~~~~~~~(c1)\\
\hspace{4em}
\includegraphics[width=2in, trim = 1.2mm 0mm 0mm 2mm, clip]{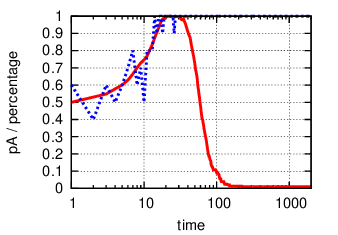}&
\includegraphics[width=2in, trim = 1.2mm 0mm 0mm 2mm, clip]{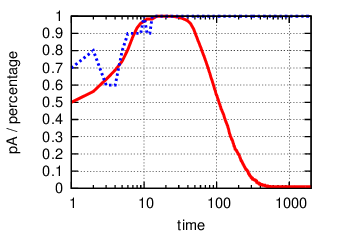}&
\includegraphics[width=2in, trim = 1.2mm 0mm 0mm 2mm, clip]{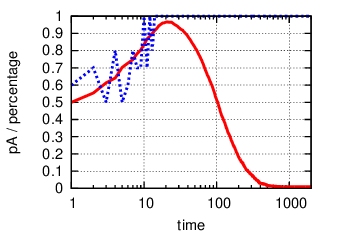}\\
~~~~~~~~(a2)&~~~~~~~~(b2)&~~~~~~~~(c2)\\
\hspace{4em}
\includegraphics[width=2in, trim = 1.2mm 0mm 0mm 2mm, clip]{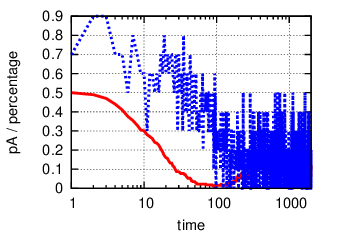}&
\includegraphics[width=2in, trim = 1.2mm 0mm 0mm 2mm, clip]{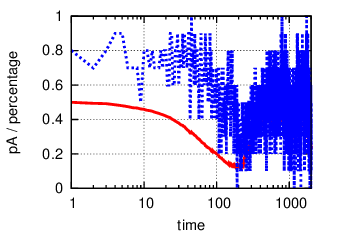}&
\includegraphics[width=2in, trim = 1.2mm 0mm 0mm 2mm, clip]{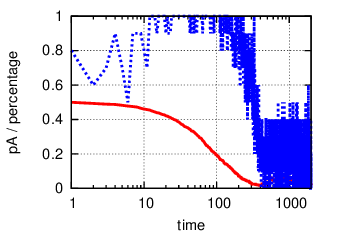}\\
~~~~~~~~(a3)&~~~~~~~~(b3)&~~~~~~~~(c3)\\
\hspace{4em}
\includegraphics[width=2in, trim = 1.2mm 0mm 0mm 2mm, clip]{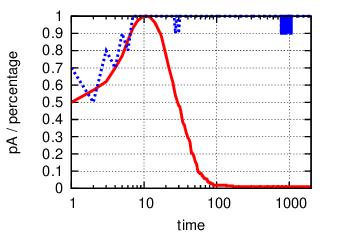}&
\includegraphics[width=2in, trim = 1.2mm 0mm 0mm 2mm, clip]{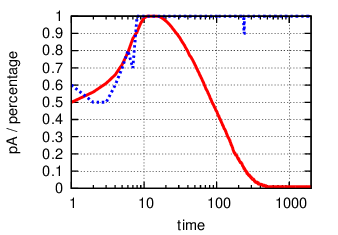}&
\includegraphics[width=2in, trim = 1.2mm 0mm 0mm 2mm, clip]{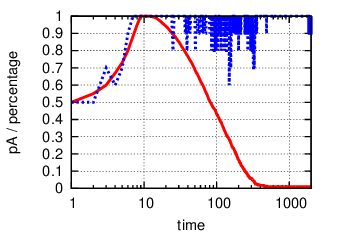}\\
~~~~~~~~(a4)&~~~~~~~~(b4)&~~~~~~~~(c4)\\
&\hspace{-18em}
\includegraphics[width=1.5in]{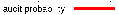}&
\hspace{-18em}
\includegraphics[width=2in]{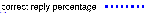}\\
\end{array}$
\end{minipage}
\caption{One covered worker. Parameters are $\Ct=0.1$, $a_i=0.1$, $\alpha=0.1$, and $\SW=0.1$ (uncovered workers) / $\SW=1$ (covered workers). Master's auditing probability, audit percentage and correct reply percentage as a function of time. First row: reputation \typeA , initial $p_C=0.5$. Second row: reputation \typeB , initial $p_C=1$.  Third row: reputation \typeC , initial $p_C=0.5$. Fourth row: reputation \typeD , initial $p_C=0.5$. First column $\tau=0.5$, $\Cp=0$. Second column, $\tau=0.1$, $\Cp=0$. Third column,  $\tau=0.1$, $\Cp=1$.}% Forth column,  $\tau=0.5$, $\Cp=1$. }
\label{cov1_sat_type1_m}
\end{figure*}

In Figure~\ref{fig4A} the correct reply rate as a function of time is presented for the case where the master covers only one worker. In a time window of 2000 rounds, as observed only reputation \typeB is able to reach eventual correctness. Even in reputation \typeD where the system looks like converged it allows the master to receive an incorrect reply. In this scenario the master has a tolerance of $\tau=0.5$ and does not punish the workers $\Cp=0$. As we can see thought in other scenarios (see Figure~\ref{cov1_sat_type1_m}) where the tolerance is minimum and the master punishes the workers, the situation remains the same; the system reaches eventual correctness only when \typeB is used. In the case of \typeA in Figure~\ref{cov1_sat_type1_m}(b1) although the master always received the correct reply, the master must always audit with $p_\VRF$ close to 1. The reason why only reputation \typeB is able to reach eventual correctness is because it is the only reputation type that fulfils Property 2. Under this property the reputation of the leading group of honest workers will never be able to surpass the reputation of the rest of the workers. Thus uncovered workers that periodically cheat only when the master does not audit will not be able to have a greater reputation than the covered workers. Thus the system will be able to always receive the correct reply from the covered worker that has the highest reputation than the rest of the workers' aggregated reputations, while reducing the auditing probability to minimum.    

%%%%%%%%%%%%%%%%%%%%%%%%%
%%%%%%%%%%%%%%%%%%%%%%%%%
%%%%%%%%%%%%%%%%%%%%%%%%%
\begin{figure*}[htbp]
%\hspace*{-1in}
 \begin{minipage}{\textwidth}
\centering
$\begin{array}{ccc}
\hspace{4.5em}
\includegraphics[width=2in, trim = 1.2mm 0mm 0mm 2mm, clip]{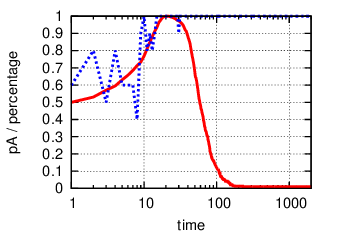}&
\includegraphics[width=2in, trim = 1.2mm 0mm 0mm 2mm, clip]{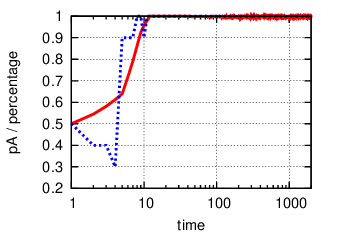}&
\includegraphics[width=2in, trim = 1.2mm 0mm 0mm 2mm, clip]{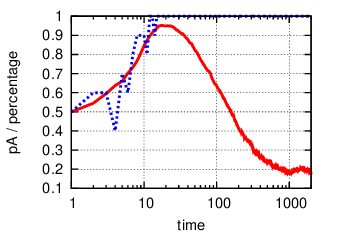}\\
~~~~~~~~(a1)&~~~~~~~~(b1)&~~~~~~~~(c1)\\
\hspace{4.5em}
\includegraphics[width=2in, trim = 1.2mm 0mm 0mm 2mm, clip]{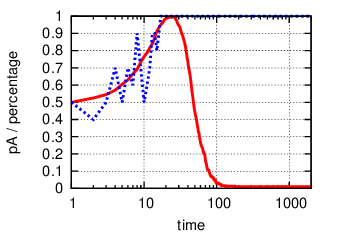}&
\includegraphics[width=2in, trim = 1.2mm 0mm 0mm 2mm, clip]{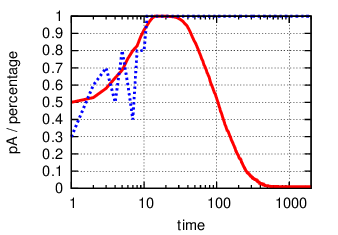}&
\includegraphics[width=2in, trim = 1.2mm 0mm 0mm 2mm, clip]{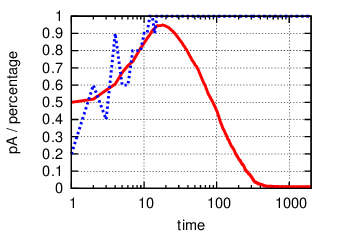}\\
~~~~~~~~(a2)&~~~~~~~~(b2)&~~~~~~~~(c2)\\
\hspace{4.5em}
\includegraphics[width=2in, trim = 1.2mm 0mm 0mm 2mm, clip]{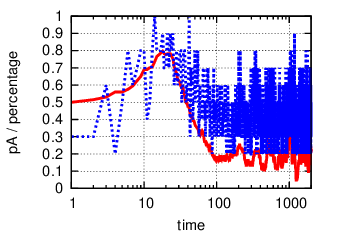}&
\includegraphics[width=2in, trim = 1.2mm 0mm 0mm 2mm, clip]{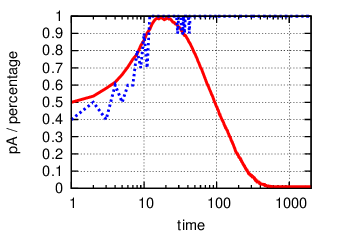}&
\includegraphics[width=2in, trim = 1.2mm 0mm 0mm 2mm, clip]{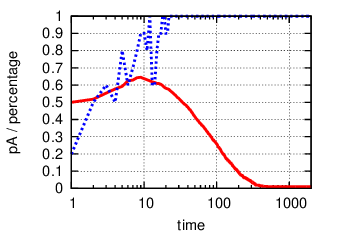}\\
~~~~~~~~(a3)&~~~~~~~~(b3)&~~~~~~~~(c3)\\
\hspace{4.5em}
\includegraphics[width=2in, trim = 1.2mm 0mm 0mm 2mm, clip]{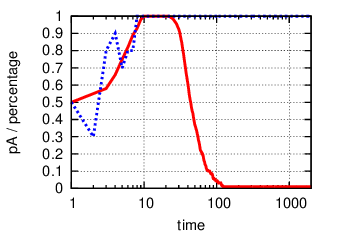}&
\includegraphics[width=2in, trim = 1.2mm 0mm 0mm 2mm, clip]{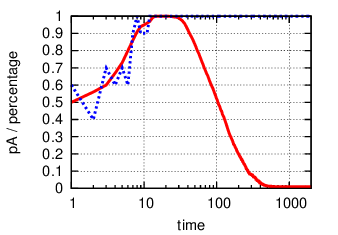}&
\includegraphics[width=2in, trim = 1.2mm 0mm 0mm 2mm, clip]{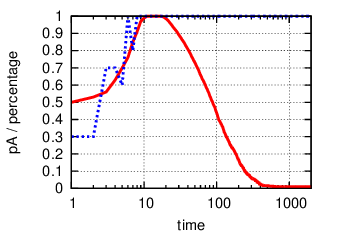}\\
~~~~~~~~(a4)&~~~~~~~~(b4)&~~~~~~~~(c4)\\
&\hspace{-18em}
\includegraphics[width=1.5in]{legends1}&
\hspace{-18em}
\includegraphics[width=2in]{legends2}\\
\end{array}$
\end{minipage}
\caption{Five covered workers. Parameters are $\Ct=0.1$, $a_i=0.1$, $\alpha=0.1$, initial $p_C=1$ and $\SW=0.1$ (uncovered workers) / $\SW=1$ (covered workers). Master's auditing probability, audit percentage and correct reply percentage as a function of time.
First row: reputation \typeA . Second row: reputation \typeB . Third row: reputation \typeC . Forth row: reputation \typeD . Left column: $\tau=0.5$, $\Cp=0$. Middle column: $\tau=0.1$, $\Cp=0$. Right column: $\tau=0.1$, $\Cp=1$.}
\label{cov5_alltype_m}
\end{figure*}

If the master covers the majority of workers then the system converges in most of the cases for different reputation types. In Figure~\ref{cov5_alltype_m} we can observe that reputation \typeB and \typeD converge in roughly the same amount of time in all the cases we examine. On the other hand reputation \typeA does not converge in the case where the tolerance is low and the master punishes (see Figure~\ref{cov5_alltype_m}(b1)). In this case the auditing probability is close to 1 meaning that this type of reputation was not sufficient to identify the covered rationals and form a trusted majority reputation forcing the master to audit almost at every round in order to obtain the correct reply. Also reputation \typeC was not able to converge although the auditing probability is less than a half the master does not always receive the correct task result. 

This lead us to conclude that is best for the master to use reputation \typeB in the case that the master can not cover the aspiration of more than one worker. If the master can cover move than the majority of workers then both reputation \typeB and \typeD are suitable. Comparing these results (see Figure~\ref{cov1_sat_type1_m}(a2) and Figure~\ref{cov5_alltype_m}(a2)\&(a4)) with the ones of Figure~\ref{pA-no-malicious}(b) we notice that the master does not need more auditing in the case of covering only a number of workers,what is perhaps counter-intuitive, thus our assumption was wrong. If the master is in a system where only rational workers exist then by using reputation \typeB she could guarantee eventual correctness by covering only one worker. Our analysis has indicated that a few bad cases exist where the system might actually not converge if not all the workers are covered even for \typeB . Although these are extreme cases as our simulations show, when a critical application exist that needs correctness of results such a risk can not be taken to reduce the costs of the master.

%\begin{figure*}[ht]
%%\hspace*{-1in}
%%\hspace*{-0.8in}
%\begin{tabular}{ccc}
%\hspace{2.5em}\includegraphics[width=2in, trim = 1mm 0mm 1.5mm 2mm, clip]{combine-pA-4-malicious-pc1-xi05-wpc0-wct01-alpha01-asp01}&
%\includegraphics[width=2in, trim = 1mm 0mm 2mm 2mm, clip]{combine-pA-5-malicious-pc1-xi05-wpc0-wct01-alpha01-asp01}&
%\includegraphics[width=2in, trim = 1mm 0mm 2mm 2mm, clip]{combine-pA-8-malicious-pc1-xi05-wpc0-wct01-alpha01-asp01}\\
%~~~~~~~~(a)&~~~~~~~~(b)&~~~~~~~~(c)
%\end{tabular}
%\caption{Master's auditing probability as a function of time in the presence of rational and malicious workers. Parameters in all plots, rationals' initial $p_C=1$, master's initial $p_\VRF=0.5$,$\SW=1$, $\Ct=0.1$, $\Cp=0$ and  $\alpha=0.1$, $a_i=0.1$. In (a) 4 malicious and 5 rationals, (b) 5 malicious and 4 rationals , (c) 8 malicious and 1 rational.}
%\label{rep_mal_rat}
%\end{figure*}
\begin{figure*}[htbp]
\begin{minipage}{\textwidth}
\centering
\subfloat[][]{\includegraphics[width=0.34\textwidth]{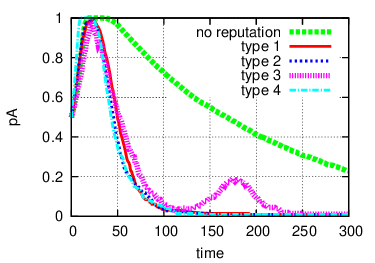}}
\subfloat[][]{\includegraphics[width=0.34\textwidth]{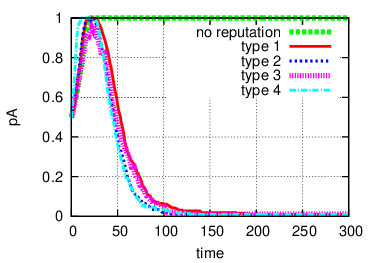}}
\subfloat[][]{\includegraphics[width=0.34\textwidth]{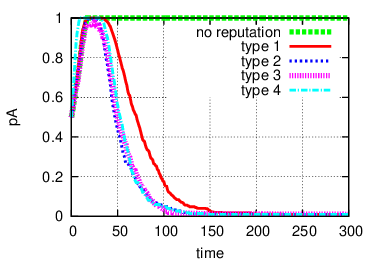}}
\end{minipage}
\caption{Master's auditing probability as a function of time in the presence of rational and malicious workers. Parameters in all plots, rationals' initial $p_C=1$, master's initial $p_\VRF=0.5$,$\SW=1$, $\Ct=0.1$, $\Cp=0$ and  $\alpha=0.1$, $a_i=0.1$. In (a) 4 malicious and 5 rationals, (b) 5 malicious and 4 rationals , (c) 8 malicious and 1 rational.}
\label{rep_mal_rat}
\end{figure*}

%\begin{figure*}[p]
%$
%\hspace{4em}
%\begin{array}{cc}
%\includegraphics[width=2in, trim = 1.2mm 0mm 1mm 2mm, clip]{combine-pA-no-rep-honest-pc1-xi05-wpc0-wct01-alpha01-asp01}&
%\includegraphics[width=2in, trim = 1.2mm 0mm 1.8mm 2mm, clip]{combine-pA-rep-t1-honest-pc1-xi05-wpc0-wct01-alpha01-asp01}\\
%~~~~~~~~~~(a)&~~~~~~~~~~(b)\\
%\includegraphics[width=2in, trim = 1.2mm 0mm 1mm 2mm, clip]{combine-pA-rep-t2-honest-pc1-xi05-wpc0-wct01-alpha01-asp01}&
%\includegraphics[width=2in, trim = 1.2mm 0mm 1.8mm 2mm, clip]{combine-pA-rep-t3-honest-pc1-xi05-wpc0-wct01-alpha01-asp01}\\
%%\includegraphics[width=1.5in, trim = 1.2mm 0mm 2.5mm 2mm, clip]{combine-pA-rep-t3-honest-pc1-xi05-wpc0-wct01-alpha01-asp01}\\
%~~~~~~~~~~(c)&~~~~~~~~~~(d)\\
%&
%\hspace{-15em}
%\includegraphics[width=2in, trim = 1.2mm 0mm 1.8mm 2mm, clip]{combine-pA-rep-t4-honest-pc1-xi05-wpc0-wct01-alpha01-asp01}\\
%&\hspace{-20em}
%~~~~~~~~~~(e)
%\end{array}$
%\caption{Master's auditing probability as a function of time in the presence of altruistic and malicious workers. Parameters in all plots, master's initial $p_\VRF=0.5$, $\Ct=0.1$, $\Cp=0$ and  $\alpha=0.1$, $a_i=0.1$. In (a) master does not use reputation, (b) master uses reputation \typeA, (c) master uses reputation \typeB , (d) master uses reputation \typeC, (e) master uses reputation \typeD .\vspace{-1em}}
%\label{rep_mal_alt}
%\end{figure*}
\begin{figure*}[htbp]
\begin{minipage}{\textwidth}
\centering
\subfloat[][]{\includegraphics[width=0.35\textwidth]{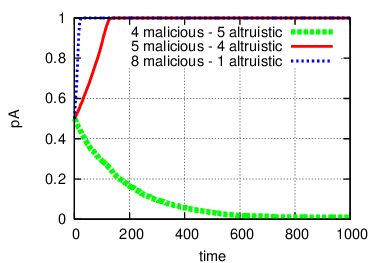}}
\subfloat[][]{\includegraphics[width=0.35\textwidth]{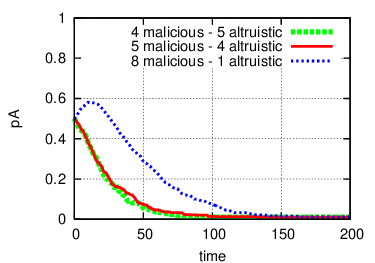}}
\end{minipage}
\begin{minipage}{\textwidth}
\centering
\subfloat[][]{\includegraphics[width=0.35\textwidth]{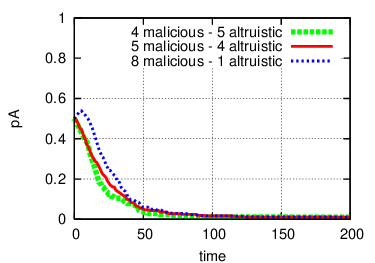}}
\subfloat[][]{\includegraphics[width=0.35\textwidth]{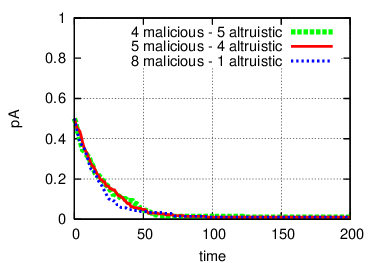}}
\subfloat[][]{\includegraphics[width=0.35\textwidth]{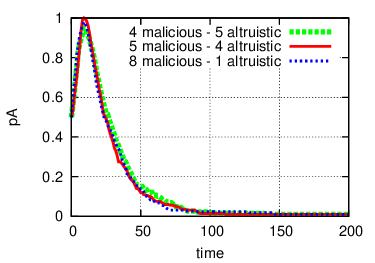}}
\end{minipage}
%\begin{minipage}{\textwidth}
%\centering
%\subfloat[][]{\includegraphics[width=0.36\textwidth]{combine-pA-rep-t4-honest-pc1-xi05-wpc0-wct01-alpha01-asp01}}
%\end{minipage}
\caption{Master's auditing probability as a function of time in the presence of altruistic and malicious workers. Parameters in all plots, master's initial $p_\VRF=0.5$, $\Ct=0.1$, $\Cp=0$ and  $\alpha=0.1$, $a_i=0.1$. In (a) master does not use reputation, (b) master uses reputation \typeA, (c) master uses reputation \typeB , (d) master uses reputation \typeC, (e) master uses reputation \typeD .}
\label{rep_mal_alt}
\end{figure*}

\subsection*{\bf\em Different types of workers} \noindent 

Moving on, we evaluate our different reputation schemes in scenarios where malicious workers exist (this was the reason for introducing reputation at the first place). 
Figure~\ref{rep_mal_rat} shows results for the extreme case, with malicious workers, no altruistic workers, and rational workers that initially cheat with probability $p_{C}=1$. We observe that if the master does not use reputation and a majority of malicious workers exist, then the master is enforced by the mechanism to audit in every round. Even with a majority of rational workers, it takes a long time for the master to reach $p_\VRF^{min}$, if reputation is not used. Introducing reputation can indeed cope with the challenge of having a majority of malicious workers, except (obviously) when all workers are malicious.
 For \typeA, the larger the number of malicious workers, the slower the master reaches $p_\VRF^{min}$. On the contrary, the time to convergence to the $p_\VRF^{min}$ is independent of the number of malicious workers for reputation \typeB . This is due to the different dynamical behavior of the two reputations as discussed before. For reputation \typeC , if a majority of rationals exists then convergence is slower. This is counter-intuitive, but it is linked to the way reputation and error rate are calculated. On the other hand, with \typeC, $p_\VRF$ is slightly lower in the first rounds. As for reputation \typeD the convergence time and the behavior of the evolution of $p_\VRF$ is similar to reputation \typeB but in the initial rounds $p_\VRF=1$ for a larger period of time, providing an additional cost to the master. 
Given the above observations we can conclude that reputation \typeB has a slight advantage in all the scenarios considered (with and without a majority of rational workers) in therms of auditing cost to the master.  
  
% We thus conclude that reputation \typeB gives better results, as long as at least one rational worker exists \ec{and the master is willing to audit slightly more in the first rounds}.

\begin{figure*}[htbp]
%\hspace*{-1in}
 \begin{minipage}{\textwidth}
\centering
\subfloat[][]{\includegraphics[width=.42\textwidth] {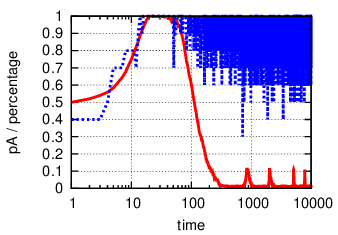}}
\qquad
\subfloat[][]{\includegraphics[width=.42\textwidth]{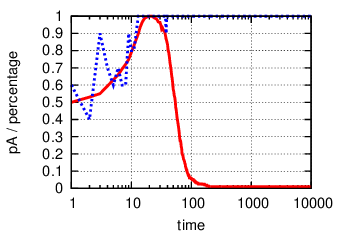} }
\end{minipage}
 \begin{minipage}{\textwidth}
\centering
\subfloat[][]{\includegraphics[width=.42\textwidth]{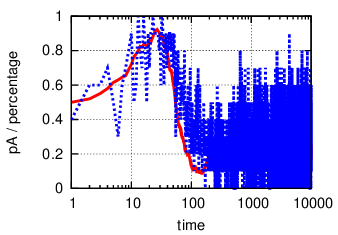} }
\qquad
\subfloat[][]{\includegraphics[width=.42\textwidth]{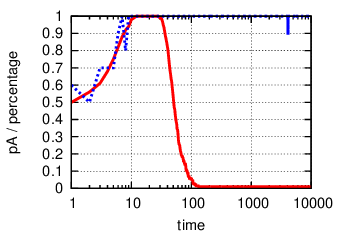} }
\end{minipage}
 \begin{minipage}{\textwidth}
\centering
\includegraphics[width=1.5in]{legends1}
\qquad
\includegraphics[width=2in]{legends2}
\end{minipage}
\caption{Presence of 4 malicious and 5 rational workers, only 1 rational worker is covered. Audit probability as function of time and correct reply percentage as a function of time. (a) Reputation \typeA . (b) Reputation \typeB . (c) Reputation \typeC . (d) Reputation \typeD . \vspace{-1em}}
\label{4mal_5rat_1covered}
\end{figure*}
 
We have checked the behavior of the system in the case where only malicious and altruistic workers exist in the system (see Figure~\ref{rep_mal_alt}). As expected, if the majority of the workers is malicious and the mechanism does not use a reputation scheme the system can not converge. For the first three reputation types the mechanism converge fast and efficiently (without increasing the initial auditing probability) for the case of 4 malicious and 5 altruistic and for the case of 5 malicious and 4 altruistic. Now for the case of 8 malicious and 1 altruistic the optimum result is given by reputation \typeC while reputation \typeB has comparably good results with a slight increment of the auditing probability in the fist rounds and convergence time around the same interval. Reputation \typeD as we can see gives the worst results with the auditing probability increasing to 1 before being able to decrease. The simulations where all three types of workers co-exist are omitted since adding altruistic workers in a system with malicious and rational workers only aids the convergence of the system without providing us with any useful inside.

In Figure~\ref{4mal_5rat_1covered}, we take a look at the case when the master decides to cover one worker out of 5 rationals and 4 malicious and that worker is rational. We notice that the system is performing in an analogous manner as in the case where all workers are covered. Only the mechanism that uses reputation \typeB is able to converge while when reputation \typeD is use the system performs quit well but still is unable to converge even after the 1000 round. 

%Finally, in Figure~\ref{4mal_5rat_1covered}, \evgenia{for the sake of experimentation we checked} that our mechanism reaches eventual correctness (with reputation \typeB) by covering only 1 out of the 5 rational workers when the other 4 are malicious. The performance of the system to reach eventual correctness is similar to the  analogous case where all workers are covered. Reputation types \typeA and \typeC have the same problems as before, whereas the fact that reputation becomes constant with \typeB  allows rational covered workers to form a reputable set by itself and achieve fast eventual correctness 

%\begin{figure}[ht]
%\begin{center}$
%\begin{array}{cc}
%\includegraphics[width=2.8in, trim = 1.1mm 0mm 1.5mm 2mm, clip]{combine-5mal-on500}
%\end{array}$
%\end{center}
%\caption{Correct reply rate as a function of time. Presence of 5 malicious workers on the 500th round. Parameters are initial $p_\VRF=0.5$,$\SW=1$, $\Ct=0.1$, $\Cp=0$ and $\alpha=0.1$, $a_i=0.1$, $\tau=0.5$, initial $p_C=1$. }
%\label{fig4B}
%\end{figure}
\begin{figure}[htbp]
\centering
\includegraphics[width=.42\textwidth]{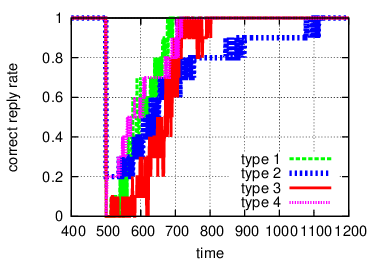}
\caption{Correct reply rate as a function of time. Presence of 5 malicious workers on the 500th round. Parameters are initial $p_\VRF=0.5$,$\SW=1$, $\Ct=0.1$, $\Cp=0$ and $\alpha=0.1$, $a_i=0.1$, $\tau=0.5$, initial $p_C=1$. }
\label{fig4B}
\end{figure}

\begin{figure*}[htbp]
%\hspace*{-1in}
$\begin{array}{ccc}
\hspace{8em}
\includegraphics[width=2.4in, trim = 1.2mm 0mm 2.5mm 2mm, clip]{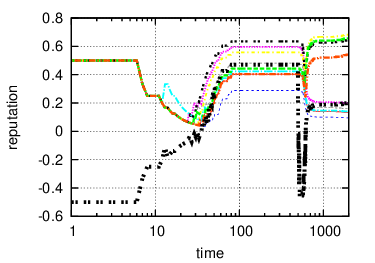}\vspace{-.2em}&
\includegraphics[width=2.4in, trim = 1.2mm 0mm 2.5mm 2mm, clip]{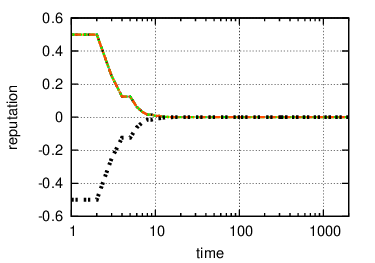}\vspace{-.2em}\\
~~~~~~~~~~~~(a)&~~~~~~~~~~(b)\\
\hspace{8em}
\includegraphics[width=2.4in, trim = 1.2mm 0mm 2.5mm 2mm, clip]{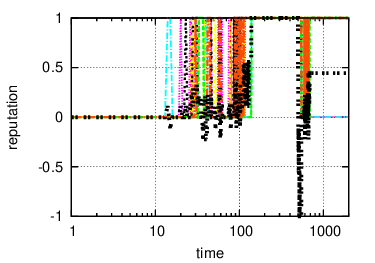}\vspace{-.2em}&
\includegraphics[width=2.4in, trim = 1.2mm 0mm 2.5mm 2mm, clip]{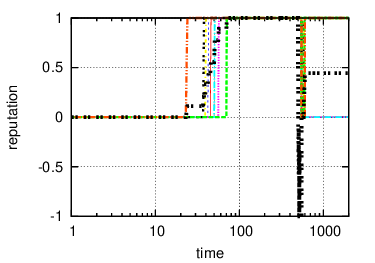}\vspace{-.1em}\\
~~~~~~~~~~~~(c)&~~~~~~~~~~(d)\\
\hspace{8em}
\includegraphics[width=1.5in]{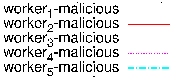}&
\hspace{8em}
\includegraphics[width=1.3in]{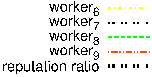}\\
\end{array}$
\caption{Presence of 5 malicious workers on the 500th round. Workers' reputation as a function of time, audit occurrences as a function of time and reputation ratio as a function of time, for an individual realization. Parameters in all panels, initial $p_C=1$, initial $p_\VRF=0.5$, $\Ct=0.1$, $\Cp=0$ and $\alpha=0.1$, $a_i=0.1$, $\tau=0.5$.
(a) Reputation \typeA , (b) reputation \typeB , (c) reputation \typeC and (d) reputation \typeD.   \vspace{-1em}}
\label{dynamic_change}
\end{figure*}

\subsection*{\bf\em Dynamic change of roles} \noindent
As a further check of the stability of our procedure, we now study the case when after correctness is reached some workers change their type, possibly due to a software or hardware error.  We simulate a situation in which  5 out of 9 rational workers suddenly change their behavior to malicious at time 500, a worst-case scenario. Figure~\ref{fig4B} shows that after the rational behavior of 5 workers turns to malicious,
convergence is reached again after a few hundred rounds and eventual correctness resumes. As we see from Figure~\ref{dynamic_change}, it takes more time for reputation \typeB to deal with the changes in the workers\rq{} behavior, 
because this reputation can never increase, and hence the system will reach eventual correctness only when the reputation of the workers that turned malicious becomes less than the reputation of the workers that stayed rational. It also takes more time for reputation \typeC to deal with the changes in the workers' behavior (see Figure~\ref{dynamic_change}).
In the case of reputation \typeA, not only the reputation of the workers that turned malicious decreases, but also the reputation of the workers that stayed rational increases. 
As for reputation \typeD it take a bit more time than reputation \typeA to reach eventual correctness again after the change of behavior. 

Therefore, reputation \typeA exhibits  better performance in dealing with dynamic changes of behavior than reputation types \typeB ,  \typeC and \typeD . As Figure~\ref{dynamic_change_master} depicts though the auditing probability of \typeA and \typeC is drastically increasing when the change of behavior happens while for reputation \typeB and \typeD this is not the case. 

%\section{Conclusions}
%In general using reputation in the presence of only rational workers would not decrease the convergence time, but %on the contrary could increase the masters auditing cost. In special worst case scenario where all workers are %rational and their initial cheating probability is one the use of reputation may have a small advantage. Using %reputation the master would decrease its auditing area and consequently, since auditing is expensive, reduce its %cost. Using a reputation mechanism in the presence of malicious workers is more efficient than having a simple %majority based mechanism and reputation type 2 is more efficient than reputation type 1. Finally the master can %reduce its auditing cost and the payment to the workers cost  if it takes into account only a number of the most %reputable workers.   

%The simulations, in the case only one worker is covered, suggest that the reputation type has to do with the %covered worker converging or not. More precisely, in the case of reputation type 2 the covered worker will %converge while in the case of reputation type 1 this is not always true.  

\section{Conclusions and Future Work}
\label{sec:conclusions}
In this work we study a malicious-tolerant generic mechanism that uses reputation. We consider four reputation types, and  give provable guarantees that only reputation \typeB  (introduced in this work) provides eventual correctness. 
%\ecc{\st{in the case of covering only one altruistic or rational worker, something that is confirmed by our simulations} 
Simulations have shown that in the case of having all rational workers covered, eventual correctness is achieved by all four types. In the case of covering only one altruistic or rational worker, simulations have shown that only reputation \typeB can achieve eventual correctness. We show that reputation \typeB has more potential in commercial platforms where high reliability together with low auditing cost, rewarding few workers and fast convergence are required. We believe this advances the development of reliable commercial Internet-based Master-Worker Computing services. 
%From our simulations we make one more interesting observation: in the case when only rational workers exist and reputation \typeC (BOINC-like) is used, although the system takes more time to converge, in every round auditing is lower. Thus,  reputation \typeC may fit better in volunteering setting where workers are most probably altruistic or rational and fast convergence can be sacrificed for lower auditing. 
In particular, our simulations reveal interesting tradeoffs between reputation types and parameters and show that our mechanism is a generic one that can be adjusted to various settings. 

The analysis of the system is done assuming that workers have an implicit form of collusion. I.e., we assume that all misbehaving workers reply with the same answer and all workers behaving correctly give the same answer. Following~\cite{anta2015algorithmic},
we are now studying stronger models, in which workers collude in deciding when to cheat and when to be honest.
In a follow-up work we plan to investigate what happens if workers are connected to each other, forming a network (i.e, a social network through which they can communicate) or if malicious workers develop a more intelligent strategy against the system. Also the degree of trust among the players has to be considered and modeled in this scenario.
Additionally, we have assumed throughout this work that workers are responsive and willing to perform the task. In a follow up work, we plan to explore the case where workers might not be responsive.
% The way we plan to deal with this problem is through selecting the most reliable workers in terms of reputation and responsiveness from a pool of workers. 
 Another extension we are planning to study is to assume a platform with multiple masters. The goal in this case is to match workers and masters to maximize social efficiency, constrained to masters' and workers' preferences. Finally, we plan to extend our mechanism to deal also with tasks where more than one responses might be considered correct.

\begin{figure*}[htbp]
%\hspace*{-1in}
$
\begin{array}{ccc}
\hspace{7em}
\includegraphics[width=2.4in, trim = 1.2mm 0mm 2.5mm 2mm, clip]{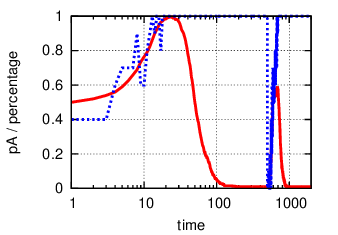}\vspace{-.2em}&
\includegraphics[width=2.4in, trim = 1.2mm 0mm 2.5mm 2mm, clip]{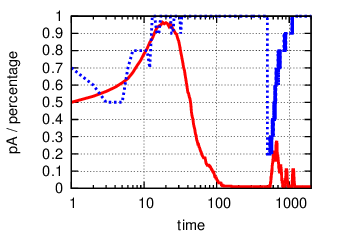}\vspace{-.2em}\\
~~~~~~~~~~~~(a)&~~~~~~~~~~(b)\\
\hspace{7em}
\includegraphics[width=2.4in, trim = 1.2mm 0mm 2.5mm 2mm, clip]{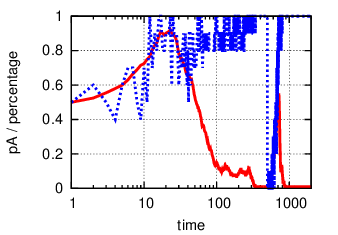}\vspace{-.2em}&
\includegraphics[width=2.4in, trim = 1.2mm 0mm 2.5mm 2mm, clip]{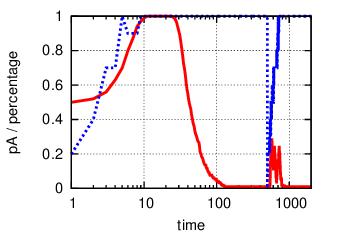}\vspace{-.2em}\\
~~~~~~~~~~~~(c)&~~~~~~~~~~(d)\\
\hspace{7em}
\includegraphics[width=1.5in]{legends1}&
\hspace{7em}
\includegraphics[width=2in]{legends2}\\
\end{array}$
\caption{Presence of 5 malicious workers on the 500th round. 
Audit probability as a function of time and correct reply percentage as a function of time. Parameters in all panels, initial $p_C=1$, initial $p_\VRF=0.5$, $\Ct=0.1$, $\Cp=0$ and $\alpha=0.1$, $a_i=0.1$, $\tau=0.5$.
(a) Reputation \typeA , (b) reputation \typeB and (c) reputation \typeC , (d) reputation \typeD .   \vspace{-1em}}
\label{dynamic_change_master}
\end{figure*}

\bibliographystyle{plain}
\bibliography{refs}

%
%\clearpage
%\begin{IEEEbiography}[{\includegraphics[width=1in,height=1.25in,clip,keepaspectratio]{evgenia}}]{Evgenia~Christoforou}
%\input{evgeniaBio}
%\end{IEEEbiography}
%\begin{IEEEbiography}[{\includegraphics[width=1in,height=1.25in,clip,keepaspectratio]{antonio}}]{Antonio~Fern\'andez~Anta}
%\input{antonioBio}
%\end{IEEEbiography}
%\begin{IEEEbiography}[{\includegraphics[width=1in,height=1.25in,clip,keepaspectratio]{chryssis}}]{Chryssis~Georgiou}
%\input{chryssisBio}
%\end{IEEEbiography}
%\begin{IEEEbiography}[{\includegraphics[width=1in,height=1.25in,clip,keepaspectratio]{miguel2015}}]{Miguel~A.~Mosteiro}
%\input{miguelBio}
%\end{IEEEbiography}
%\begin{IEEEbiography}[{\includegraphics[width=1in,height=1.25in,clip,keepaspectratio]{anxoPic}}]{Angel~S\'anchez}
%\input{anxoBio}
%\end{IEEEbiography}
%
%\EOD

\end{document}